\documentclass[11pt,a4paper]{article}
\usepackage[english]{babel} 
\usepackage{amsmath}
\usepackage{mathtools}
\usepackage{amsfonts}
\usepackage{amssymb}
\usepackage{amsthm}
\usepackage{makeidx}
\usepackage{graphicx}
\usepackage{natbib}
\usepackage{setspace}
\usepackage{float}
\usepackage{hyperref}
\usepackage[margin=1in]{geometry}
\usepackage{setspace}
\usepackage[titletoc]{appendix}
\hypersetup{
	colorlinks   = true,
	citecolor    = blue,
	urlcolor = blue
}

\author{Marcelo Brutti Righi\footnote{We would like to thank Professors Ruodu Wang and Martin Schweizer for their comments, which have helped to improve the	manuscript. We are grateful for the financial support of FAPERGS (Rio Grande do Sul State Research Council) project number 17/2551-0000862-6 and CNPq (Brazilian Research Council) projects number 302369/2018-0 and 407556/2018-4.}
	\\
	\textit{Federal University of Rio Grande do Sul}\\
	marcelo.righi@ufrgs.br}

\title{A theory for combinations of risk measures}
\date{}

\newtheorem{Def}{Definition}[section]
\newtheorem{Thm}[Def]{Theorem}
\newtheorem{Prp}[Def]{Proposition}
\newtheorem{Lmm}[Def]{Lemma}
\newtheorem{Crl}[Def]{Corollary}
\theoremstyle{definition}
\newtheorem{Exm}[Def]{Example}
\newtheorem{Asp}[Def]{Assumption}
\theoremstyle{remark}
\newtheorem{Rmk}[Def]{Remark}

\numberwithin{equation}{section}

\DeclareMathOperator*{\argmax}{arg\,max}

\begin{document}
	
\maketitle
\begin{abstract}
	We study combinations of risk measures under no restrictive assumption on the set of alternatives. We develop and discuss results regarding the preservation of properties and acceptance sets for the combinations of risk measures. One of the main results is the representation of resulting risk measures from the properties of both alternative functionals and combination functions. We build on developing a dual representation for an arbitrary mixture of convex risk measures. In this case, we obtain a penalty that recalls the notion of inf-convolution under theoretical measure integration. We develop results related to this specific context. We also explore features of individual interest generated by our frameworks, such as the preservation of continuity properties and the representation of worst-case risk measures.
\end{abstract}	

\section{Introduction}
The theory of risk measures in mathematical finance has become mainstream, especially since the landmark paper of \cite{Artzner1999}. For a comprehensive review, see the books of \cite{Pflug2007}, \cite{Delbaen2012}, \cite{Ruschendorf2013} and \cite{Follmer2016}. Nonetheless, there is still no consensus about the best theoretical properties to possess, and even less regarding the best risk measure. See \cite{Emmer2015} for a comparison of risk measures. This phenomenon motivates the proposition of new approaches, such as in \cite{Righi2016} and \cite{Righi2020}, for instance. Under the lack of a versatile choice of the best risk measure from a set of alternatives, one can consider the common use of many candidates to benefit from distinct qualities. 

However, such a choice can lead to multidimensional or even infinite dimensional problems that
bring complexity that may make the treatment of risk measurement impossible to handle. For example, in a portfolio optimization problem, to take into account this variety of features from different risk measures, the agent may end up with a very complex multi-objective function or even an exorbitant volume of constraints. Such a situation would lead to elevated computational costs or even the impossibility of a feasible solution. Hence, the alternative is to consider a combination of all the candidates instead of all them individually. In the context of portfolio optimization, we would have a single constraint or objective with a larger feasibility space.

The drawback of considering such a combination is that we may end up without the main characteristics that define risk measures. More precisely, the axiomatic theory of risk measures strongly relies on a set of financial and mathematical properties related to dual representations and acceptance sets. Thus, understanding how to preserve such properties in a general fashion is crucial to guarantee the usefulness of combinations. In this sense, developing a theoretical body for combinations of risk measures is pivotal. A general approach must deal with arbitrary sets of candidates, even uncountable ones. Such a situation may arise when the parameter that defines the candidates relies on a subset of the real line, such as the significance level for Value at Risk or some probability of default in credit risk. 

The main challenge relies on the generality needed to perform this kind of task since we cannot rely on methods for finite-dimensional spaces that appear in the literature. For instance, when dealing with some uncountable set of candidates, we cannot even consider the usual summation, which is crucial for averaging, having to replace it with integration. However, in this case, we have measurability issues to take into account that may become complex. Also, set operations may not preserve topological properties, such as the uncountable union of closed sets does not have to be closed. Even the choice of suitable domains for combination functions in infinite dimensional spaces can be, per se, a source of complexity since there is no canonical functional space. 

Another source of difficulty and complexity is that combinations may assume any functional form, such as averaging or supremum-based worst cases. Any function applied over the set of candidate risk measures can be considered. Thus, studying the impact of such combinations in properties, acceptance sets, and dual representations, among other features, is not straightforward. Since specific combination functions can be more suitable for distinct contexts, having a general treatment is very beneficial. In this sense, a theory that provides a reliable and practical manner to preserve desired properties and obtain acceptance sets or dual representations for a general combination of risk measures can help improve other fields in mathematical finance.

Under this background, in this paper we study risk measures of the form $\rho=f(\rho_{\mathcal{I}})$, where $\rho_{\mathcal{I}}=\{\rho^i,\:i\in\mathcal{I}\}$ is a set of alternative risk measures and $f$ is some combination function. We propose a framework whereby no assumption is made on the index set $\mathcal{I}$, apart from non-emptiness. Typically, this procedure uses a finite set of candidates, leading the domain of $f$ to be some Euclidean space. In our case, the domain of $f$ is taken by a subset of the random variables over a suitable measurable space created on $\mathcal{I}$. From that, our main goal is to establish general results on properties, develop dual representations, and study acceptance sets for such composed risk measures based on the properties of both $\rho_{\mathcal{I}}$ and $f$ in a general sense. For this purpose, we expose results for some featured special cases, which are also of particular interest, such as a worst case and mixtures of risk measures.

There are studies regarding particular cases for $f$, $\mathcal{I}$ and $\rho_{\mathcal{I}}$, such as the worst case in \cite{Follmer2002}, the sum of monetary and deviation measures in \cite{Righi2018a}, finite convex combinations in \cite{Ang2018}, scenario-based aggregation in \cite{Wang2018}, model risk-based weighting over a non-additive measure in \cite{Jokhadze2018}. Our main contribution is that we do not restrict the set of alternative risk measures and the generality of combination functions we consider. Furthermore, none of such papers considers all the features we take into account. It is worthy of mentioning that the well-known concept of inf-convolution of risk measures as in \cite{Barrieu2005} or \cite{Jouini2008} is not suitable
for the approach in this paper since even in
the most simple case of two risk measures, and we cannot write it as a direct
combination as $f(\rho_1(X),\rho_2(X)) = \inf_Y \{\rho_1(X-Y)+\rho_2(X)\}$ since it is not only a function of $X$. More precisely, the inf-convolution
depends of every allocation $X_1 + X_2 = X$. The work of \cite{Righi2020b} is focused on inf-convolution and optimal risk sharing for arbitrary sets of risk
measures.

We have structured the rest of this paper as follows: in Section \ref{sec:back} we expose preliminaries regarding notation, a brief background on the theory of risk measures in order to support our framework and our proposed approach with some examples; in Section \ref{sec:prop} we present results regarding properties of combination functions and how they affect the resulting risk measures in both financial and continuity properties; in Section \ref{sec:main} we develop and prove our results on representations of resulting risk measures in terms of properties from both the set of candidates and the combination for the general convex and law invariant cases, as well we address a representation for the worst-case risk measure; in Section \ref{sec:accept} we present results on how properties of combination functions affect the resulting risk measures acceptance set and provide a general characterization for the convex and coherent cases.

	\section{Preliminaries}\label{sec:back}
	
	\subsection{Notation}
	
	Consider the probability space $(\Omega,\mathcal{F},\mathbb{P})$. All equalities and inequalities are in the $\mathbb{P}$-a.s. sense. We have that $L^0=L^0(\Omega,\mathcal{F},\mathbb{P})$ and $L^{\infty}=L^{\infty}(\Omega,\mathcal{F},\mathbb{P})$ are, respectively, the spaces of (equivalent classes under $\mathbb{P}$-a.s. equality of) finite and essentially bounded random variables. We define $1_A$ as the indicator function for an event $A\in\mathcal{F}$. We identify constant random variables with real numbers. We say that a pair $X,Y\in L^0$ is comonotone if $\left(X(w)-X(w^{\prime})\right)\left( Y(w)-Y(w^{\prime}) \right)\geq0$ holds $\mathbb{P}\times\mathbb{P}$-a.s. We denote by $X_n\rightarrow X$ convergence in the $L^\infty$ essential supremum norm $\lVert \cdot\rVert_{\infty}$, while $\lim\limits_{n\rightarrow\infty}X_n=X$ means $\mathbb{P}$-a.s. convergence. Let $\mathcal{P}$ be the set of all probability measures on $(\Omega,\mathcal{F})$. We denote $E_{\mathbb{Q}}[X]=\int_{\Omega}Xd\mathbb{Q}$, $F_{X, \mathbb{Q}}(x)=\mathbb{Q}(X\leq x)$ and $F_{X, \mathbb{Q}}^{-1}(\alpha)=\inf\left\lbrace x:F_{X, \mathbb{Q}}(x)\geq\alpha\right\rbrace $, respectively, the expected value, the (non-decreasing and right-continuous) probability function and its inverse for $X$ under $\mathbb{Q}\in\mathcal{P}$. We write $X\overset{\mathbb{Q}}\sim Y$ when $F_{X,\mathbb{Q}}=F_{Y,\mathbb{Q}}$. We drop subscripts regarding probability measures when $\mathbb{Q}=\mathbb{P}$. Furthermore, let $\mathcal{Q}\subseteq\mathcal{P}$ be the set of probability measures that are continuous concerning $\mathbb{P}$ with Radon-Nikodym derivatives $\frac{d\mathbb{Q}}{d\mathbb{P}}$. 
	
	\subsection{Background}
	
	We begin with the definition of risk measures and the properties they may or not fulfill. We focus here on the most used properties in the literature. For a detailed interpretation, we refer to the books cited above.
	
	\begin{Def}\label{def:risk}
		A functional $\rho:L^\infty\rightarrow\mathbb{R}$ is a risk measure. Its acceptance set is defined as $\mathcal{A}_\rho=\left\lbrace X\in L^\infty:\rho(X)\leq 0 \right\rbrace $. $\rho$ may possess the following properties: 
		
		\begin{enumerate}
			\item Monotonicity [M]: if $X \leq Y$, then $\rho(X) \geq \rho(Y),\:\forall\: X,Y\in L^\infty$.
			\item Translation Invariance [TI]: $\rho(X+C)=\rho(X)-C,\:\forall\: X,Y\in L^\infty,\:\forall\:C \in \mathbb{R}$.
			\item Convexity [C]: $\rho(\lambda X+(1-\lambda)Y)\leq \lambda \rho(X)+(1-\lambda)\rho(Y),\:\forall\: X,Y\in L^\infty,\:\forall\:\lambda\in[0,1]$.
			\item Positive Homogeneity [PH]: $\rho(\lambda X)=\lambda \rho(X),\:\forall\: X,Y\in L^\infty,\:\forall\:\lambda \geq 0$.
			\item Law Invariance [LI]: if $F_X=F_Y$, then $\rho(X)=\rho(Y),\:\forall\:X,Y\in L^\infty$.
			\item Comonotonic Additivity [CA]: $\rho(X+Y)= \rho(X)+\rho(Y),\:\forall\: X,Y\in L^\infty$ with $X,Y$ comonotone.
			\item Fatou continuity [FC]:  if $\lim\limits_{n\rightarrow\infty}X_n=X\in L^\infty$ and $\{X_n\}_{n=1}^\infty\subseteq L^\infty$ bounded, then $\rho(X) \leq \liminf\limits_{n\rightarrow\infty} \rho( X_{n})$.
		\end{enumerate}
		
		We have that  $\rho$ is called monetary if it fulfills [M] and [TI], convex if it is monetary and respects [C], coherent if it is convex and fulfills [PH], law invariant if it has [LI], comonotone if it attends [CA], and Fatou continuous if it possesses [FC]. In this paper, we are working with normalized risk measures in the sense of $\rho(0)=0$. 
	\end{Def}
	
	Beyond usual norm and Fatou-based continuities, (a.s.) point-wise are relevant for risk measures since they play a role in representations. Since their preservation is also important, we now define the most used ones in the literature.
	
	\begin{Def}\label{def:cont}
		A risk measure $\rho:L^\infty\rightarrow\mathbb{R}$ is said to be:
		\begin{enumerate}
			\item Continuous from above: $\lim\limits_{n\rightarrow\infty}X_n=X$ and $\{X_n\}$ non-increasing implies in $\rho(X)= \lim\limits_{n\rightarrow\infty} \rho( X_{n})$, $\:\forall\:\{X_n\}_{n=1}^\infty,X\in L^\infty$.
			\item Continuous from below: $\lim\limits_{n\rightarrow\infty}X_n=X$ and $\{X_n\}$ non-decreasing  implies in $\rho(X)= \lim\limits_{n\rightarrow\infty} \rho( X_{n})$, $\:\forall\:\{X_n\}_{n=1}^\infty,X\in L^\infty$.
			\item Lebesgue continuous: $\lim\limits_{n\rightarrow\infty}X_n=X$ implies in $\rho(X)= \lim\limits_{n\rightarrow\infty} \rho( X_{n})$ $\:\forall\:\{X_n\}_{n=1}^\infty\subseteq L^\infty$ bounded and $\forall\:X\in L^\infty$.
		\end{enumerate}
	\end{Def}
	
	For convex risk measures, [FC] is equivalent to continuity from above, while [FC] is implied by continuity from below. Thus, for convex risk measures, continuity from below is equivalent to Lebesgue continuity. We recommend the books mentioned in the classic theory to review details regarding interpreting such properties. 
	
	The main core of the risk measures theory is the triplet composed of properties, acceptance sets, and dual representations. Since this is the main focus of this paper, we need the following results for direct implications regarding the interactions of these three structures. In fact, 	under [C] and [FC], the well-known convex duality plays an important role. 
	
	\begin{Thm}[Proposition 4.6 in \cite{Follmer2016}]\label{Thm:accept}
		Let $\mathcal{A}_\rho$ be the acceptance set defined by $\rho\colon L^\infty\rightarrow\mathbb{R}$. Then:
		\begin{enumerate}
			\item If $\rho$ fulfills [M], then $X\in\mathcal{A}_\rho$, $Y\in L^\infty$ and $Y\geq X$ implies in $Y\in\mathcal{A}_\rho$. In particular, $L^\infty_+\subseteq\mathcal{A}_\rho$.
			\item If $\rho$ fulfills [TI], then $\rho(X)=\inf\left\lbrace m\in\mathbb{R}\colon X+m\in\mathcal{A}_\rho\right\rbrace$.
			\item If $\rho$ is a monetary risk measure, then $A_\rho$ is non-empty, closed with respect to the supremum norm, $\mathcal{A}_\rho\cap\{X\in L^\infty\colon X<0\}=\emptyset$, and $\inf\{m\in\mathbb{R}\colon m\in\mathcal{A}_\rho\}>-\infty$.
			\item If $\rho$ attends [C], then $\mathcal{A}_\rho$ is a convex set.
			\item If $\rho$ fulfills [PH], then $\mathcal{A}_\rho$ is a cone.
		\end{enumerate}
	\end{Thm}

	\begin{Thm}[Theorem 2.3 of \cite{Delbaen2002}, Theorem 4.33 of \cite{Follmer2016}]\label{the:dual}
		Let $\rho : L^\infty\rightarrow \mathbb{R}$ be a risk measure. Then:
		\begin{enumerate}
			\item $\rho$ is a Fatou continuous convex risk measure if and only if it can be represented as:
			\begin{equation}\label{eq:dual}
			\rho(X)=\sup\limits_{\mathbb{Q}\in\mathcal{Q}}\left\lbrace E_\mathbb{Q}[-X]-\alpha^{min}_\rho(\mathbb{Q}) \right\rbrace,\:\forall\:X\in L^\infty,
			\end{equation}
			where  $\alpha^{min}_\rho : \mathcal{Q}\rightarrow\mathbb{R}_+\cup\{\infty\}$, defined as $\alpha^{min}_\rho(\mathbb{Q})=\sup\limits_{X\in\mathcal{A}_\rho}E_\mathbb{Q}[-X]$, is a lower semi-continuous convex function called penalty term. This is equivalent to $\mathcal{A}_\rho$ be weak* closed, i.e. $\sigma(L^\infty,L^1)$ closed.
			\item 	$\rho$ is a Fatou continuous coherent risk measure if and only if it can be represented as:
			\begin{equation}\label{eq:cohdual}
			\rho(X)=\sup\limits_{\mathbb{Q}\in\mathcal{Q}_\rho} E_\mathbb{Q}[-X],\:\forall\:X\in L^\infty,
			\end{equation} where $\mathcal{Q}_\rho\subseteq\mathcal{Q}$ is non-empty, closed and convex  called the dual set of $\rho$.
		\end{enumerate}
	\end{Thm}

	\begin{Exm}\label{Exm:meas}
		Examples of risk measures:
		\begin{enumerate}
			\item Expected Loss (EL): This is a Fatou continuous law invariant comonotone coherent risk measure defined as $EL(X)=-E[X]=-\int_0^1 F^{-1}_X(s)ds$. We have that $\mathcal{A}_{EL}=\left\lbrace X\in L^\infty:E[X]\geq0\right\rbrace $ and $\mathcal{Q}_{EL}=\{\mathbb{P}\}$. 
			\item Value at Risk (VaR): This is a Fatou continuous law invariant comonotone monetary risk measure defined as $VaR^\alpha(X)=-F_{X}^{-1}(\alpha),\:\alpha\in[0,1]$. We have  that $\mathcal{A}_{VaR^\alpha}=\left\lbrace X\in L^\infty:\mathbb{P}(X<0)\leq\alpha\right\rbrace $.
			\item Expected Shortfall (ES): This is a Fatou continuous law invariant comonotone coherent risk measure  defined as $ES^{\alpha}(X)=\frac{1}{\alpha}\int_0^\alpha VaR^s(X)ds,\:\alpha\in(0,1]$ and $ES^0(X)=VaR^0(X)=-\operatorname{ess}\inf X$. We have $\mathcal{A}_{ES^\alpha}=\left\lbrace X\in L^\infty:\int_0^\alpha VaR^s(X)ds\leq0\right\rbrace $ and  $\mathcal{Q}_{ES^\alpha}=\left\lbrace \mathbb{Q}\in\mathcal{Q} : \frac{d\mathbb{Q}}{d\mathbb{P}}\leq\frac{1}{\alpha} \right\rbrace$.
			\item Maximum loss (ML): This is a Fatou continuous law invariant  coherent risk measure defined as $ML(X)=-\text{ess}\inf X=F^{-1}_X(0)$. We have $\mathcal{A}_{ML}=\left\lbrace X\in L^\infty:X\geq0\right\rbrace$ and $\mathcal{Q}_{ML}=\mathcal{Q}$.
		\end{enumerate}
	\end{Exm}

	Interesting features are present when there is [LI], which is the case in most practical applications. In this situation, we always assume that the probability space is atom-less.
	
	\begin{Thm}[Theorem 2.1 of \cite{Jouini2006} and Proposition 1.1 of \cite{Svindland2010}]\label{thm:fatou}
		Let $\rho:L^\infty\rightarrow\mathbb{R}$ be a law invariant convex risk measure. Then $\rho$ is Fatou continuous.
	\end{Thm}
	
	\begin{Thm}[Theorems 4 and 7 of \cite{Kusuoka2001}, Theorem 4.1 of \cite{Acerbi2002a}, Theorem 7 of \cite{Fritelli2005}]\label{thm:dual2}
		Let $\rho : L^\infty\rightarrow \mathbb{R}$ be a risk measure. Then:
		\begin{enumerate}
			\item $\rho$ is a law invariant convex risk measure if and only if it can be represented as:
			\begin{equation}\label{eq:dual2}
			\rho(X)=\sup\limits_{m\in\mathcal{M}}\left\lbrace \int_{(0,1]}ES^\alpha(X)dm-\beta^{min}_{\rho}(m)\right\rbrace,\:\forall\:X\in L^\infty,
			\end{equation}
			where $\mathcal{M}$ is the set of probability measures on $(0,1]$ and $\beta^{min}_{\rho} : \mathcal{M}\rightarrow\mathbb{R}_+\cup\{\infty\}$, defined as
			$\beta^{min}_{\rho}(m)=\sup\limits_{X\in\mathcal{A}_{\rho}} \int_{(0,1]}ES^\alpha(X)dm$. 
			\item $\rho$ is a law invariant coherent risk measure if and only if it can be represented as:
			\begin{equation}\label{eq:dual2coh}
			\rho(X)=\sup\limits_{m\in\mathcal{M}_\rho} \int_{(0,1]}ES^\alpha(X)dm,\:\forall\:X\in L^\infty,
			\end{equation}
			where  $\mathcal{M}_{\rho}=\left\lbrace m\in\mathcal{M}\colon\int_{(u,1]}\frac{1}{v}dm=F^{-1}_{\frac{d\mathbb{Q}}{d\mathbb{P}}}(1-u),\:\mathbb{Q}\in\mathcal{Q}_\rho\right\rbrace $.
			\item $\rho$ is a law invariant comonotone coherent risk measure if and only if it can be represented as:
			\begin{equation}\label{eq:dual2com}
			\rho(X)= \int_{(0,1]}ES^\alpha(X)dm,\:\forall\:X\in L^\infty,
			\end{equation}
			where $m\in\mathcal{M}_\rho$.  
		\end{enumerate}
	\end{Thm}

	\subsection{Proposed approach}
	
	Let $\rho_\mathcal{I}=\{\rho^i\colon L^\infty\rightarrow\mathbb{R},\:i\in\mathcal{I}\}$ be some (a priori specified) collection of risk measures, where $\mathcal{I}$ is a non-empty set. We write, for fixed $X\in L^\infty$, $\rho_\mathcal{I}(X)=\{\rho^i(X),\:i\in\mathcal{I}\}$. We would like to define risk measures as $\rho(X)=f(\rho_\mathcal{I}(X))$, where $f$ is some combination (aggregation) function. When $\mathcal{I}$ is finite with dimension $n$, we have that $f\colon\mathbb{R}^n\rightarrow\mathbb{R}$. This situation, which is common in practical matters, simplifies the framework. However, as the introduction exposes, when $\mathcal{I}$ is an arbitrary set, we need a more complex setup. Examples of such complexity are that, in this case, we may have to deal with integration and measurability, some set operations may not preserve topological properties, and the absence of a canonical functional space to be the domain of $f$.

	Consider the measurable space $(\mathcal{I},\mathcal{G})$. We define  $K^0=K^0(\mathcal{I},\mathcal{G})$ and $K^\infty=K^\infty(\mathcal{I},\mathcal{G})$ as the spaces of  finite and bounded  random variables, respectively. In these spaces, we understand equalities, inequalities, and limits in the point-wise sense. We define $\mathcal{V}$ as the set of finitely additive measures $\mu$ in $(\mathcal{I},\mathcal{G})$ such that $\mu(\mathcal{I})=1$. With some abuse of notation, we write $\int_\mathcal{I} fd\mu$ for the usual bi-linear form. In order to avoid measurability issues, we make the following assumption.
	
	\begin{Asp}\label{asp:measure}
		The maps $R_X\colon\mathcal{I}\rightarrow\mathbb{R}$, defined as $R_X(i)=\rho^i(X)$, are $\mathcal{G}$-measurable for any $X\in L^\infty$.   
	\end{Asp}
	
	\begin{Rmk}
		A possible, but not unique, choice is when $\mathcal{G}=\sigma(\{i\rightarrow\rho^i(X)\colon X\in L^\infty\})=\sigma(\{R_X^{-1}(B)\colon B\in\mathcal{B}(\mathbb{R}),X\in L^\infty\})$, where $\mathcal{B}(\mathbb{R})$ is the Borel set of $\mathbb{R}$. Another possibility is, of course, the power set $2^\mathcal{I}$. A situation of interest is when $\mathcal{G}$ is a Borel sigma-algebra and $i\rightarrow\rho^i(X)$ continuous for any $X$. Other case of interest is $\mathcal{I}=[0,1]$ and $\rho^i$ composed by $VaR^i$ or $ES^i$, for instance.
	\end{Rmk}
	
	We can associate the domain of $f$ with $\mathcal{X}=\mathcal{X}{\rho_\mathcal{\mathcal{I}}}=span(\{R\in K^0\colon\exists\: X\in L^\infty\; \text{s.t}.\; R(i)=\rho^i(X),\:\forall\:i\in\mathcal{I}\}\cup\{1\})=span(\{R\in K^0\colon R=R_X,X\in L^\infty\}\cup\{1\})$. The linear span is in order to preserve vector space operations. We can identify $\rho_\mathcal{I}(X)$ with $R_X$ by $\rho_{\mathcal{I}}\colon L^\infty\rightarrow\mathcal{X}$. From normalization, we have $R_0=0$. When $\rho_\mathcal{I}(X)$ is bounded, which is the case for any monetary risk measure since $\rho(X)\leq\rho(\operatorname{ess}\inf X)=-\operatorname{ess}\inf X<\infty$, similarly for $\rho(X)\geq-\operatorname{ess}\sup X>-\infty$, we have that $\mathcal{X}\subseteq K^\infty$. Under this framework, the composition is a functional $f\colon\mathcal{X}\rightarrow\mathbb{R}$. We use, when necessary, the canonical extension convention that $f(R)=\infty$ for $R\in K^0\backslash\mathcal{X}$. We consider normalized combination functions as $f(R_0)=f(0)=0$.
	
	\begin{Exm}\label{Ex:WC}
		The worst-case risk measure is a functional $\rho^{WC}:L^\infty\rightarrow\mathbb{R}$ defined as
		\begin{equation}\label{eq:WC}
		\rho^{WC}(X)=\sup\limits_{i\in\mathcal{I}}\rho^i(X).
		\end{equation}
		This risk measure is typically considered when the agent (investor, regulator, etc.) seeks protection. When $\mathcal{I}$ is finite, the supremum is, of course, a maximum. This combination is the point-wise supremum $f^{WC}(R)=\sup\{R(i)\colon\:i\in\mathcal{I}\}$. If $\mathcal{X}\subseteq K^\infty$, then $\rho^{WC}<\infty$. When  $\mathcal{I}=\mathcal{Q}$ and $\rho^\mathbb{Q}(X)=E_\mathbb{Q}[-X]-\alpha(\mathbb{Q})$, with $\alpha\colon\mathcal{Q}\rightarrow\mathbb{R}_+\cup\{\infty\}$ such that $\inf\{\alpha(\mathbb{Q})\colon\mathbb{Q}\in\mathcal{Q}\}=0$, we have that $\rho^{WC}$ becomes a Fatou continuous convex risk measure as \eqref{eq:dual} in Theorem \ref{the:dual}.  Analogously, for a non-empty closed convex $\mathcal{I}\subseteq\mathcal{Q}$ and $\rho^\mathbb{Q}(X)=E_\mathbb{Q}[-X]$, we have that $\rho^{WC}$ becomes a Fatou continuous coherent risk measure as \eqref{eq:cohdual} in this same Theorem \ref{the:dual}. Analogous analysis can be made to obtain law invariant convex and coherent risk measures as \eqref{eq:dual2coh} and \eqref{eq:dual2}, respectively, as in Theorem \ref{thm:dual2}.
	\end{Exm} 
	\begin{Exm}\label{ex:mu}
		The weighted risk measure is a functional $\rho^\mu:L^\infty\rightarrow\mathbb{R}$ defined as
		\begin{equation}\label{eq:mu}
		\rho^\mu(X)=\int_{\mathcal{I}}\rho^{i}(X)d\mu,
		\end{equation}
		where $\mu\in\mathcal{V}$. This risk measure represents an expectation of $R_X$ regarding $\mu$. Since $i\rightarrow\rho^i(X)$ is $\mathcal{G}$-measurable, the integral is well-defined. In addition, it is finite when $\mathcal{X}\in K^\infty$. When $\mathcal{I}$ is finite, $\rho^\mu$ is a convex mixture of the functionals which compose $\rho_\mathcal{I}$. The combination function is $f^\mu(R)=\int_{\mathcal{I}}Rd\mu$. We have that $|\rho^{\mu_1}(X)-\rho^{\mu_2}(X)|\leq \lvert\rho^{WC}(X)\rvert\lVert\mu_1-\mu_2\rVert_{TV}$, where $\lVert\cdot\rVert_{TV}$ is the total variation norm. Hence it is somehow continuous (robust) regarding the choice of the probability measure $\mu\in\mathcal{V}$. If $\mathcal{I}=(0,1]$ and $\rho^i(X)=ES^i(X)$, we have that $\rho^\mu(X)$ defines a law invariant comonotone convex risk measure as \eqref{eq:dual2com}, which is Fatou continuous due to Theorem \ref{thm:fatou}.  
	\end{Exm}

	\begin{Exm}\label{ex:spec}
		
		A spectral (distortion) risk measure is a functional $\rho^\phi\colon L^\infty\rightarrow\mathbb{R}$ defined as
		\begin{equation}
		\rho^\phi(X)=\int_0^1VaR^\alpha(X)\phi(\alpha)d\alpha,
		\end{equation}
		where $\phi:[0,1]\rightarrow\mathbb{R}_+$ is a non-increasing functional such that $\int_{0}^{1}\phi(u)du=1$. Any law invariant comonotone convex risk measure can be expressed in this fashion. The relationship between this representation and the one in \eqref{eq:dual2com} is given by $\int_{(u,1]}\frac{1}{v}dm=\phi(u)$, where $m\in\mathcal{M}_\rho$. When  $\phi$ is not non-increasing, we have that the risk measure is not convex, and the representation as combinations of ES does not hold. Let  $\mathcal{I}=[0,1]$, $\lambda$ the Lebesgue measure, and $\mu\ll\lambda$ with $\phi(i)=F^{-1}_{\frac{d\mu}{d\lambda}}(1-i)$. Thus $\int_\mathcal{I}\phi d\lambda=1$ and $\rho^\phi(X)=\int_{\mathcal{I}}\rho^i(X)\phi(i)d\lambda$. By choosing $\rho^i(X)=VaR^i(X)$, we have that any spectral risk measure is a special case of $\rho^\mu$. 
		
	\end{Exm}
	
	\begin{Exm}\label{ex:ut}
		Consider the risk measure $\rho^{u}(X)\colon L^\infty\rightarrow\mathbb{R}$ defined as
		\begin{equation}\label{eq:uti}
		\rho^u(X)=u(\rho_{\mathcal{I}}(X)),
		\end{equation}
		where $u:\mathcal{X}\rightarrow\mathbb{R}$ is a monetary utility in the sense that if $R\geq S$, then $u(R)\geq u(S)$ and $u(R+C)=u(R)+C,\:C\in\mathbb{R}$. In this case, the combination is $f^u=u$. Note that $u(R)$ can be identified with $\pi(-R)$, where $\pi$ is a risk measure on $\mathcal{X}$. For instance, one can pick $\pi$ as EL, VaR, ES, or ML. In these cases we would obtain for some base probability $\mu$, respectively the following combinations: $f^\mu$, $F^{-1}_{R,\mu}(1-\alpha)$, $\frac{1}{\alpha}\int_0^\alpha F^{-1}_{R,\mu}(1-s)ds$, and $\operatorname{ess}\sup_\mu R$. 
	\end{Exm}
	
	\begin{Exm}
		For this example we denote  $\mathbb{F}=\{F_{X,\mathbb{Q}}\colon\:X\in L^\infty,\:\mathbb{Q}\in\mathcal{P}\}$. In this case, uncertainty is linked to probabilities in the sense that  $\mathcal{I}\subseteq\mathcal{P}$ and we can define risk measurement under the intuitive idea that we obtain the same functional from distinct probabilities that represent scenarios. We then have a probability-based risk measurement as a family of risk measures $\rho_\mathcal{I}=\{\rho^\mathbb{Q}: L^\infty\rightarrow \mathbb{R},\:\mathbb{Q}\in\mathcal{I}\}$ such that \begin{equation}
		\rho^\mathbb{Q}(X)=\mathcal{R}^\rho(F_{X,\mathbb{Q}}),\:\forall\:X\in L^\infty,\:\forall\:\mathbb{Q}\in\mathcal{I},
		\end{equation}where $\mathcal{R}^\rho:\mathbb{F}\rightarrow\mathbb{R}$ is called risk functional. In this context, $f$ can be any statistical map of interest to deal with this probabilistic effect, such as median and mode. This setup can also be utilized to deal with probabilistic issues and uncertainty. In fact, under a suitable choice for $f$ connected to dispersion, it is possible to use our approach to quantify model uncertainty.
		
	\end{Exm}

	\section{Properties}\label{sec:prop}
	
	\subsection{Properties of combinations}
	In this section, we expose results regarding the preservation of properties for composed risk measures $\rho=f(\rho_{\mathcal{I}})$ based on the properties of both $\rho_{\mathcal{I}}$ and $f$. We begin by defining the properties of the composition $f$.
	
	\begin{Def}\label{def:comp}
		A combination $f:\mathcal{X}\rightarrow\mathbb{R}$ may have the following properties:
		\begin{enumerate}
			\item Monotonicity [M]: if $R\geq S$, then $f(R)\geq f(S),\:\forall\:R,S\in\mathcal{X}$.
			\item Translation Invariance [TI]: $f(R+C)=f(R)+C,\:\forall\:R,S\in\mathcal{X},\:\forall\:C\in\mathbb{R}$.
			\item Positive Homogeneity [PH]: $f(\lambda S)=\lambda f(S),\:\forall\:R\in\mathcal{X},\:\forall\:\lambda\geq0$.
			\item Convexity [C]: $f(\lambda R+(1-\lambda)S)\leq\lambda f(R)+(1-\lambda)f(S),\:\forall\:\lambda\in[0,1],\:\forall\:R,S\in\mathcal{X}$.
			\item Additivity [A]: $f(R+S)= f(R)+f(S),\:\forall\:R,S\in\mathcal{X}$.
			\item Fatou continuity [FC]: If $\lim\limits_{n\rightarrow\infty}R_n=R\in K^\infty$, with $\{R_n\}_{n=1}^\infty\subseteq\mathcal{X}$ bounded, then $f(R)\leq\liminf\limits_{n\rightarrow\infty}f(R_n)$.
		\end{enumerate}
	\end{Def}
	
	\begin{Rmk}
		Such properties for the combination function $f$ are parallel to those of risk measures, exposed in Definition \ref{def:risk}. Note the adjustment in signs from there. We use the same terms indiscriminately for $f$ and $\rho$ with reasoning to fit the context. We could have imposed a determined set of properties for the combination. However, we choose to keep a more general framework where it may or may not possess such properties.
	\end{Rmk}

	\begin{Prp}\label{prp:fmu}
		Let $\mathcal{X}\subseteq K^\infty$. We have that:
		\begin{enumerate}
			\item $f^{WC}$ defined as in Example \ref{Ex:WC} fulfills [M], [TI], [PH], [C] and [FC].
			\item $f^\mu$ defined as in Example \ref{ex:mu} fulfills [M], [TI], [PH], [C], [A] and [FC].
		\end{enumerate}
	\end{Prp}
	
	\begin{proof}
		\begin{enumerate}	
			\item	Properties [M], [TI], [PH] and [C] are obtained directly from the definition of supremum. Regarding [FC], let $\{R_{n}\}_{n=1}^\infty\subseteq\mathcal{X}$ bounded such that $\lim\limits_{n\rightarrow\infty}R_{n}=R\in K^\infty$. Then we have that \[f^{WC}(R)=\sup\lim\limits_{n\rightarrow\infty}R_{n}\leq\liminf\limits_{n\rightarrow\infty}\sup R_{n}=\liminf\limits_{n\rightarrow\infty}f^{WC}(R_{n}).\] 
			\item 	From the properties of integration, we have that $f^\mu$ respects [M], [TI], [PH], [C] and [A]. For [FC], let  $\{R_{n}\}_{n=1}^\infty\subseteq\mathcal{X}$ bounded such that $\lim\limits_{n\rightarrow\infty}R_{n}=R\in K^\infty$. Then we have from Dominated Convergence that  \[f^\mu(R)=\int_{\mathcal{I}}\lim\limits_{n\rightarrow\infty}R_{n}d\mu
			\leq\liminf\limits_{n\rightarrow\infty}\int_{\mathcal{I}}R_{n}d\mu=\liminf\limits_{n\rightarrow\infty}f^\mu(R_{n}).\]
		\end{enumerate}
	\end{proof}
	
	\begin{Rmk}
		Note that for any combination $f$ with the property of Boundedness, i.e  $|f(R)|\leq f^{WC}(R),\:\forall\:R\in\mathcal{X}$, we have $\rho(X)\leq\rho^{WC}(X)$. Consequently $\mathcal{A}_{\rho^{WC}}\subseteq\mathcal{A}_\rho$. From Theorem \ref{the:dual} applied to functionals over $K^\infty$, we have that $\{f^\mu\}_{\mu\in\mathcal{V}}$ are the only combination functions that fulfill all properties in Definition \ref{def:comp}. 
	\end{Rmk}
	
	\subsection{Financial properties}
	
	We now focus on the preservation of financial properties.
	
	\begin{Prp}\label{prp:prop}
		Let $\rho_\mathcal{I}=\{\rho^i\colon L^\infty\rightarrow\mathbb{R},\:i\in\mathcal{I}\}$ be a collection of risk measures, $f\colon\mathcal{X}\rightarrow\mathbb{R}$, and $\rho\colon L^\infty\rightarrow\mathbb{R}$ a risk measure defined as $\rho(X)=f(\rho_\mathcal{I}(X))$. Then:
		\begin{enumerate}
			\item If $\rho_\mathcal{I}$ is composed of risk measures with [M], and $f$ possesses this same property, then also does $\rho$.
			\item If $\rho_\mathcal{I}$ is composed of risk measures with [TI], and $f$ possesses this same property, then also does $\rho$.
			\item If $\rho_\mathcal{I}$ is composed of risk measures with [C] and $f$ possesses this same property in pair with [M], then $\rho$ fulfills [C].
			\item If $\rho_\mathcal{I}$ is composed of risk measures with [PH], and $f$ possesses this same property, then also does $\rho$.
			\item If $\rho_\mathcal{I}$ is composed of law invariant risk measures, then  $\rho$ fulfills [LI].
			\item If $\rho_\mathcal{I}$ is composed of comonotone risk measures and $f$ fulfills [A], then  $\rho$ has [CA].
			\item If $\rho_\mathcal{I}$ is composed of Fatou continuous point-wise bounded risk measures and $f$ has [FC] in pair with [M], then also does $\rho$.
		\end{enumerate}
	\end{Prp}
	
	\begin{proof}
		
		\begin{enumerate}
			\item Let $X,Y\in L^\infty$ with $X\geq Y$. Then $\rho^i(X)\leq\rho^i(Y),\:\forall\:i\in\mathcal{I}$. Thus, $R_X\leq R_Y$ and $\rho(X)=f(R_X)\leq f(R_Y)=\rho(Y)$.
			\item Let $X\in L^\infty$ and $C\in\mathbb{R}$. Then $\rho^i(X+C)=\rho^i(X)-C,\:\forall\:i\in\mathcal{I}$. Thus, $\rho(X+C)=f(R_{X+C})=f(R_X-C)=f(R_X)-C=\rho(X)-C$.
			\item Let $X,Y\in L^\infty$ and $\lambda\in[0,1]$. Then $\rho^i(\lambda X+(1-\lambda)Y)\leq \lambda \rho^i(X)+(1-\lambda)\rho^i(Y),\:\forall\:i\in\mathcal{I}$. Thus, $\rho(\lambda X+(1-\lambda Y))=f(R_{\lambda X+(1-\lambda Y)})\leq f(\lambda R_X+(1-\lambda)R_Y)\leq\lambda\rho(X)+(1-\lambda)\rho(Y)$.
			\item Let $X\in L^\infty$ and $\lambda\geq0$. Then $\rho^i(\lambda X=\lambda \rho^i(X),\:\forall\:i\in\mathcal{I}$. Thus $\rho(\lambda X)=f(R_{\lambda X})=f(\lambda R_X)=\lambda f(R_X)=\lambda\rho(X)$.
			\item Let $X,Y\in L^\infty$ such that $F_X=F_Y$. Then $\rho^i(X)=\rho^i(Y),\forall\:i\in\mathcal{I}$. Thus $R_X=R_Y$ point-wisely. Hence, $R_X$ and $R_Y$ belong to the same equivalence class on $\mathcal{X}$ and $\rho(X)=f(R_X)=f(R_Y)=\rho(Y)$. 
			\item Let $X,Y\in L^\infty$ be a comonotone pair. Then $\rho^i(X+Y)=\rho^i(X)+\rho^i(Y),\:\forall\:i\in\mathcal{I}$. Thus $\rho(X+Y)=f(R_{X+Y})=f(R_X+R_Y)=f(R_X)+f(R_Y)=\rho(X)+\rho(Y)$.
			\item Let$\{X_n\}_{n=1}^\infty\subseteq L^\infty$ bounded with $\lim\limits_{n\rightarrow\infty}X_n=X\in L^\infty$. Then $\rho^i(X) \leq \liminf\limits_{n\rightarrow\infty} \rho^i( X_{n}),\:\forall\:i\in\mathcal{I}$. Since $\sup_{n\in\mathbb{N}}\rho_{\mathcal{I}}(X_n)\leq\sup_{n\in\mathbb{N}}\lVert X_n\rVert_\infty <\infty$, we get that $\{R_{X_n}\}$ is bounded. Thus $\rho(X)=f(R_X)\leq f(\liminf\limits_{n\rightarrow\infty}R_{X_n})\leq\liminf\limits_{n\rightarrow\infty}f(R_{X_n})=\liminf\limits_{n\rightarrow\infty}\rho(X_n)$.
		\end{enumerate}
	\end{proof}
	
	\begin{Rmk}
		Converse relations are not always guaranteed. For instance, spectral risk measures in Example \ref{ex:spec} are convex despite the collection $\{VaR^\alpha,\alpha\in[0,1]\}$ is not in general. Moreover, The preservation of Subadditivity [SA], i.e. $\rho(X+Y)\leq \rho(X)+\rho(Y),\:\forall\:X,Y\in L^\infty$, is quite similar to the one for [C], obtained by replacing the properties. The same for Comonotone Convexity and Comonotone Subadditivity, see \cite{Kou2013}, which are relaxed counterparts for Comonotonic pairs.
	\end{Rmk}
	
	For completeness, we now investigate the preservation of other properties in the literature on risk measures. 
	
	\begin{Prp}\label{prp:prop3}
		Let $\rho_\mathcal{I}=\{\rho^i\colon L^\infty\rightarrow\mathbb{R},\:i\in\mathcal{I}\}$ be a collection of risk measures, $f\colon\mathcal{X}\rightarrow\mathbb{R}$, and $\rho\colon L^\infty\rightarrow\mathbb{R}$ a risk measure defined as $\rho(X)=f(\rho_\mathcal{I}(X))$. Then:
		\begin{enumerate}
			\item If $\rho_\mathcal{I}$ is composed of risk measures with [C] and $f$ possesses [M] and Quasi-convexity [QC], i.e. $f(\lambda R+(1-\lambda)S)\leq\max\{f(R),f(S)\},\:\forall\:\lambda\in[0,1],\:\forall\:R,S\in \mathcal{X}$, then $\rho$ fulfills [QC].
			\item If  $\rho_\mathcal{I}$ is composed of risk measures with Cash-subadditivity [CS], i.e. $\rho^i(X+C)\geq \rho^i(X)-C,\:\forall\:C\in\mathbb{R}_+,\:\forall\:X\in L^\infty$, and $f$ possesses [M] and [TI], then $\rho$ has [CS].
			\item  If $\rho_\mathcal{I}$ is composed of risk measures with Relevance [R], i.e. $X\leq0$ and $\mathbb{P}(X<0)>0$ imply $\rho^i(X)>0,\:\forall\:X\in L^\infty$, and $f$ has strict [M], i.e. $R>S$ imply $f(R)>f(S),\:\forall\:R,S\in\mathcal{X}$, then $\rho$ has [R].
			\item If $\rho_\mathcal{I}$ is composed of risk measures with Surplus Invariance [SI], i.e. $\rho^i(X)\leq 0$ and $Y^-\leq X^-$ imply $\rho^i(Y)\leq 0,\:\forall\:X,Y\in L^\infty$, and $f$ has [M] together to $f\geq f^{WC}$, then $\rho$ possesses [SI]. 
		\end{enumerate}
	\end{Prp}
	
	\begin{proof}
		\begin{enumerate}
			\item Let $X,Y\in L^\infty$ and $\lambda\in[0,1]$. Then $\rho(\lambda X+(1-\lambda)Y)=f(R_{\lambda X+(1-\lambda Y)})\leq f(\lambda R_X+(1-\lambda)R_Y)\leq\max\{\rho(X),\rho(Y)\}$.
			\item Take $X\in L^\infty$ and $C\in\mathbb{R}_+$. Then $\rho(X+C)\geq f(R_X-C)=f(R_X)-C=\rho(X)-C$.
			\item Let $X\in L^\infty$ such that  $X\leq0$ and $\mathbb{P}(X<0)>0$. Then $\rho_\mathcal{I}(X)>0$ point-wisely in $\mathcal{X}$. Thus, from strict Monotonocity of $f$ we get $\rho(X)=f(R_X)>0$.
			\item Let $X,Y\in L^\infty$ such that $\rho(X)\leq 0$ and $Y^-\leq X^-$. Thus, from hypotheses we get $\rho^i(X)\leq f^{WC}(R_X)\leq f(R_X)\leq 0$ for any $i\in\mathcal{I}$. This leads to $\rho_{\mathcal{I}}(Y)\leq 0$  point-wisely in $\mathcal{X}$. Hence, $\rho(Y)=f(R_Y)\leq 0$.
		\end{enumerate}
	\end{proof}
	
	\begin{Rmk}\label{rmk:prop}
		See \cite{Cerreia2011}, \cite{ElKaroui2009}, \cite{Delbaen2012} and \cite{Koch-Medina2017} for  details on [QC], [CS], [R] and [SI], respectively. Unfortunately, [QC] of $\rho_{\mathcal{I}}$ is not preserved if $f$ is monotone and quasi-convex (or even convex). The proof of such a claim relies on the fact that $\mathcal{X}$ does not have a total order, in contrast to the case of the real line where $f\circ\rho$ would be quasi-convex. Moreover, even in the case exposed in item (i) of the last Proposition, in order to guarantee preservation for [CS], which is the typical companion for [QC], we must assume that $f$ possesses [TI]. Since it would imply [C], we would be back at the original framework of the paper.
	\end{Rmk}
	
	\subsection{Continuity properties}
	
	In this subsection, the goal is to preserve continuity properties beyond [FC].
	
	\begin{Prp}\label{prp:WCcont}
		Let $\rho_\mathcal{I}=\{\rho^i\colon L^\infty\rightarrow\mathbb{R},\:i\in\mathcal{I}\}$ be a collection of point-wise bounded risk measures, $f\colon\mathcal{X}\rightarrow\mathbb{R}$, and $\rho\colon L^\infty\rightarrow\mathbb{R}$ a risk measure defined as $\rho(X)=f(\rho_\mathcal{I}(X))$.  Then:
		\begin{enumerate}
			\item If  $\rho_\mathcal{I}$ is composed of Lipschitz continuous risk measures, and $f$ fulfills [M], [SA], and Boundedness, then $\rho$ is Lipschitz continuous. 
			\item If  $\rho_\mathcal{I}$ is composed of risk measures that possess continuity from above, below, or Lebesgue and $f$ is Fatou continuous, then $\rho$ is Fatou continuous for non-increasing sequences, non-decreasing sequences, or any sequences, respectively.
			\item If $\rho_\mathcal{I}$ is composed of risk measures that posses any property among continuity from above, below or Lebesgue, and $f$ is Lebesgue continuous, i.e. $\lim\limits_{n\rightarrow\infty}R_n=R$ implies $f(R)=\lim\limits_{n\rightarrow\infty}f(R_n)$, $\forall\:\{R_n\}_{n=1}^\infty\subseteq\mathcal{X}$ bounded and any $R\in\mathcal{X}\cap K^\infty$, then $\rho$ also does.
		\end{enumerate}
	\end{Prp}
	
	\begin{proof}
		
		\begin{enumerate}
			\item From [M] and [SA] of $f$ and $R_{X}\leq R_Y+|R_{X}-R_Y|$ we have that $|f(R_{X})-f(R_Y)|\leq f(|R_{X}-R_Y|)$. Moreover, $\lvert \rho^i(X)-\rho^i(Y)\rvert\leq \lVert X-Y\rVert_\infty,\:\forall\:X,Y\in L^\infty,\:\forall\:i\in\mathcal{I}$. Then from the Boundedness of $f$, we get
			\[	\left|  \rho(X)-\rho(Y)\right| \leq f(\left|  R_{X}-R_Y\right|)\leq f^{WC}(\left|  R_{X}-R_Y\right|) \leq\lVert X-Y\rVert_\infty,\:\forall\:X,Y\in L^\infty.\]
			\item Let  $\{X_n\}_{n=1}^\infty\subseteq L^\infty$ bounded such that $\lim\limits_{n\rightarrow\infty}X_{n}=X\in L^\infty$ and $\rho^i$ Lebesgue continuous for any $i\in\mathcal{I}$. Then we have $\{R_{X_n}\}$ bounded and $\lim\limits_{n\rightarrow\infty}R_{X_n}=R_X$ point-wise. Hence
			\[\rho(X)=f\left(\lim\limits_{n\rightarrow\infty}R_{X_n}\right)\leq\liminf\limits_{n\rightarrow\infty}f(R_{X_n})=\liminf\limits_{n\rightarrow\infty}\rho(X_n).\] When each $\rho^i$ is continuous from above or below, the same reasoning which is restricted to non-decreasing or non-increasing sequences, respectively, is valid.
			\item Similar to (ii), but in this case $f\left(\lim\limits_{n\rightarrow\infty}R_{X_n}\right)=\lim\limits_{n\rightarrow\infty}f(R_{X_n})$.	
			
		\end{enumerate}
	\end{proof}
	
	\begin{Rmk}
		Item (i) can be generalized to $\rho_{\mathcal{I}}$ uniformly equicontinuous, i.e. the $\delta-\epsilon$ criteria does not depend on $i\in\mathcal{I}$, from which Lipschitz continuity is a special case. Moreover, let $\mathcal{I}=\mathcal{Q}$, $\rho^\mathbb{Q}(X)=E_\mathbb{Q}[-X]$ and $f=f^{WC}$. In this case, we have $\rho=ML$, which is not continuous from below even  $\rho^\mathbb{Q}$ possessing such property. However, ML is Fatou continuous. This example illustrates item (ii) in the last Proposition. Moreover, $f^\mu$ satisfies Lebesgue continuity as in (iii) when $\mathcal{X}\subseteq K^\infty$, which is the case for monetary risk measures, for instance.
	\end{Rmk}

	\section{Representations}\label{sec:main}
\subsection{General result}

In this section, we expose results regarding the representation of composed risk measures $\rho=f(\rho_{\mathcal{I}})$ based on the properties of both $\rho_{\mathcal{I}}$ and $f$. The goal is to highlight the role of such terms. We begin the preparation with a lemma for representation of $f$, without dependence on the properties of $\rho_{\mathcal{I}}$.

\begin{Lmm}\label{lmm:f}
	Let $\mathcal{X}\subseteq K^\infty$. A functional $f\colon \mathcal{X}\rightarrow\mathbb{R}$,  posses [M], [TI] and [C] if and only if it can be represented as
	\begin{equation}\label{eq:convf}
	f(R)=\sup\limits_{\mu\in\mathcal{V}}\left\lbrace\int_{\mathcal{I}}Rd\mu-\gamma_f(\mu) \right\rbrace,\:\forall\:R\in\mathcal{X},
	\end{equation}
	where $\gamma_f:\mathcal{V}\rightarrow\mathbb{R}_+\cup\{\infty\}$ is defined as
	\begin{equation}\label{eq:penf}
	\gamma_f(\mu)=\sup\limits_{R\in\mathcal{X}}\left\lbrace\int_{\mathcal{I}}Rd\mu-f(R) \right\rbrace.
	\end{equation}
\end{Lmm}
\begin{proof}
	The fact that \eqref{eq:convf} possesses [M], [TI] and [C]  is straightforward. For the only if direction, one can understand $f(R)$ as $\pi(-R)$, where $\pi$ is a convex risk measure on $K^\infty$. Note that it is finite. Thus, from Theorem 4.16 in \cite{Follmer2016}  applied to $K^\infty$, and $f(R)=\infty$ for any $R\in K^\infty\backslash\mathcal{X}$ we have that \[f(R)=\pi(-R)=\sup\limits_{\mu\in\mathcal{V}}\left\lbrace\int_{\mathcal{I}}Rd\mu-\sup\limits_{R\in \mathcal{X}}\left[\int_{\mathcal{I}}Rd\mu-f(R) \right] \right\rbrace=\sup\limits_{\mu\in\mathcal{V}}\left\lbrace\int_{\mathcal{I}}Rd\mu-\gamma_f(\mu) \right\rbrace.\]
	
\end{proof}

\begin{Rmk}\label{Rmk:dualf}
	When $R=R_X$ point-wisely for some $X\in L^\infty$, we have that the representation becomes \[	f(R_X)=\sup\limits_{\mu\in\mathcal{V}}\left\lbrace\rho^\mu(X)-\gamma_f(\mu) \right\rbrace.\] If $f$ possesses [PH], then $\gamma_f$ assumes value $0$ in $\mathcal{V}_f=\{\mu\in\mathcal{V}\colon f(R)\geq\int_{\mathcal{I}}Rd\mu,\:\forall\:R\in \mathcal{X}\}$ and $\infty$ otherwise. For instance, $\mathcal{V}_{f^\mu}=\{\mu\}$ and $\mathcal{V}_{f^{WC}}=\mathcal{V}$. Note that $\inf\limits_{\mu\in\mathcal{V}}\gamma_f(\mu)=0$ from the assumption of normalization for $f$.
	
\end{Rmk}

We need the following auxiliary result, which may be of individual interest regarding the integration of probability measures.

\begin{Lmm}\label{Lmm:probs}
	Let $\{\mathbb{Q}^i,i\in\mathcal{I}\}$ such that $i\rightarrow\mathbb{Q}^i(A)$ is $\mathcal{G}$-measurable  for any $A\in\mathcal{F}$. Then $\mathbb{Q}(A)=\int_{\mathcal{I}}\mathbb{Q}^i(A)d\mu,\:\forall\:A\in\mathcal{F},\mathbb{Q}^i\in\mathcal{Q}\:\mu-a.s.$ defines a probability measure. In this case, $E_{\mathbb{Q}}[X]=\int_{\mathcal{I}}E_{\mathbb{Q}^i}[X]d\mu$ for any $X\in L^\infty$.
\end{Lmm}

\begin{proof}
	Let  $\mathbb{Q}(A)=\int_{\mathcal{I}}\mathbb{Q}^i(A)d\mu,\:\forall\:A\in\mathcal{F}$. It is direct that both $\mathbb{Q}(\emptyset)=0$ and $\mathbb{Q}(\Omega)=1$. For countable additivity, let $\{A_n\}_{n\in\mathbb{N}}$ be a collection of mutually disjoint sets. Then, since $i\rightarrow\mathbb{Q}^i(A)$ is bounded $\forall\:A\in\mathcal{F}$ we have \[\mathbb{Q}(\cup_{n=1}^\infty A_n)=\int_{\mathcal{I}}\sum_{n=1}^{\infty}\mathbb{Q}^i(A_n)d\mu=\sum_{n=1}^{\infty}\int_{\mathcal{I}}\mathbb{Q}^i(A_n)d\mu=\sum_{n=1}^{\infty}\mathbb{Q}(A_n).\] Hence $\mathbb{Q}$ is a probability measure. Regarding expectation interchange we have for any $X\in L^\infty$ that $x\rightarrow \mathbb{Q}^i(X\leq x)=F_{X,\mathbb{Q}^i}(x)$ is monotone and right-continuous $i\in\mathcal{I}$. Then $(x,i)\rightarrow\mathbb{Q}^i(X\leq x)$ is $\mathcal{B}(\mathbb{R})\otimes\mathcal{G}$-measurable, indeed integrable. Hence we have \begin{align*}
	E_{\mathbb{Q}}[-X]&=\int_{0}^{\infty}\int_{\mathcal{I}}(1-\mathbb{Q}^i(X\leq x))d\mu dx+\int_{-\infty}^{0}\int_{\mathcal{I}}\mathbb{Q}^i(X\leq x)d\mu dx\\
	&=\int_{\mathcal{I}}\left( \int_{0}^{\infty}(1-\mathbb{Q}^i(X\leq x))dx+\int_{-\infty}^{0}\mathbb{Q}^i(X\leq x)dx\right) d\mu\\
	&=\int_{\mathcal{I}}E_{\mathbb{Q}^i}[-X]d\mu.
	\end{align*}
	By changing signs, we get the claim.
\end{proof}

We also need an assumption to circumvent some measurability issues to avoid the indefiniteness of posterior measure-related concepts, such as integration. \begin{Asp}\label{Lmm:measure}When $\rho_\mathcal{I}=\{\rho^i\colon L^\infty\rightarrow\mathbb{R},\:i\in\mathcal{I}\}$ is a collection of Fatou continuous convex risk measures we assume that  $i\rightarrow\alpha^{min}_{\rho^{i}}(\mathbb{Q})$ is $\mathcal{G}$-measurable for any $\mathbb{Q}\in\mathcal{Q}$.\end{Asp}

\begin{Rmk} Similarly to Assumption \ref{asp:measure}, as a single, but not unique, example for $\mathcal{G}$, one could consider the sigma-algebra generated by all such maps. Again, a situation of interest is when $\mathcal{G}$ is a Borel sigma-algebra and $i\rightarrow\alpha^{\min}_{\rho^i}(\mathbb{Q})$ continuous for any $\mathbb{Q}$. This is the case when $\mathcal{I}=[0,1]$ and $\rho^i$ composed by $ES^i$, or even when $\mathcal{I}=(0,\infty)$ and $\rho^i$ entropic risk measures under penalty $\alpha^{\min}_{\rho^i}(\mathbb{Q})=\frac{1}{i}E\left[\log\left(\frac{d\mathbb{Q}}{d\mathbb{P}} \right) \frac{d\mathbb{Q}}{d\mathbb{P}} \right]$.

\end{Rmk}

The role played by $\rho^\mu$ becomes clear since it can be understood as the expectation under $\mu$ of elements $R_X\in\mathcal{X}$. Thus, it is important to know how the properties of $\rho_\mathcal{I}$ affect the representation of $\rho^\mu$.   Proposition 2.1 of \cite{Ang2018} explores a case with a finite number of coherent risk measures while we address a situation with an arbitrary set of convex risk measures.

\begin{Thm}\label{lmm:dualmu}
	Let $\rho_\mathcal{I}=\{\rho^i\colon L^\infty\rightarrow\mathbb{R},\:i\in\mathcal{I}\}$ be a collection of Fatou continuous convex risk measures and $\rho^{\mu}: L^\infty\rightarrow \mathbb{R}$ defined as in \eqref{eq:mu}. Then:
	\begin{enumerate}
		\item  $\rho^\mu$ can be represented as:
		\begin{equation}\label{eq:convmu}
		\rho^\mu(X)=\sup\limits_{\mathbb{Q}\in\mathcal{Q}}\left\lbrace E_\mathbb{Q}[-X]-\alpha_{\rho^\mu}(\mathbb{Q}) \right\rbrace,\:\forall\:X\in L^\infty,
		\end{equation}
		with   $\alpha_{\rho^\mu}\colon\mathcal{Q}\rightarrow\mathbb{R}_+\cup\{\infty\}$ defined as 
		\begin{equation}\label{eq:pennu}
		\alpha_{\rho^\mu}(\mathbb{Q})=
		\inf\left\lbrace\begin{cases*}
		\int_{\mathcal{I}}\alpha^{min}_{\rho^i}\left(\mathbb{Q}^i\right)d\mu,\text{if}\:i\to\alpha^{min}_{\rho^i}\left(\mathbb{Q}^i\right)\in\mathcal{G}\\
		\infty,\:\text{otherwise}
		\end{cases*} \colon \int_{\mathcal{I}}\mathbb{Q}^id\mu=\mathbb{Q},\:\mathbb{Q}^i\in\mathcal{Q}\:\forall\:i\in\mathcal{I}\right\rbrace .
		\end{equation}
		
		\item If in addition $\rho^i$ fulfills, for every $i\in\mathcal{I}$, [PH], then the representation is 
		\begin{equation}\label{eq:cohmu}
		\rho^\mu(X)=\sup\limits_{\mathbb{Q}\in cl\left(\mathcal{Q}_{\rho^\mu}\right)}E_\mathbb{Q}[-X],\:\forall\:X\in L^\infty,
		\end{equation}
		with $\mathcal{Q}_{\rho^\mu}=\left\lbrace \mathbb{Q}\in\mathcal{Q}\colon \mathbb{Q}=\int_{\mathcal{I}}\mathbb{Q}^id\mu,\mathbb{Q}^i\in\mathcal{Q}_{\rho^i}\:\forall\:i\in\mathcal{I}\right\rbrace $ convex and non-empty, where $cl$ means closure in total variation norm. 
	\end{enumerate}
\end{Thm}

\begin{proof}
	
	\begin{enumerate}
		\item  Note that $\alpha_{\rho^\mu}$ is well-defined since the infimum is not altered for the distinct choices of possible combinations that lead to $\int_{\mathcal{I}}\mathbb{Q}^id\mu=\mathbb{Q}$. 	Non-negativity for $\alpha_{\rho^\mu}$ is straightforward. Also, from Assumption \ref{Lmm:measure}, we have that $\alpha_{\rho^\mu}(\mathbb{Q})\leq\int_{\mathcal{I}}\alpha^{min}_{\rho^i}\left(\mathbb{Q}\right)d\mu,\:\forall\:\mathbb{Q}\in\mathcal{Q}$.	The measurability of $i\rightarrow\rho^i(X)=\sup_{\mathbb{Q}\in\mathcal{Q}}\left\lbrace E_\mathbb{Q}[-X]-\alpha^{\min}_{\rho^i}(\mathbb{Q}) \right\rbrace $ for any $X\in L^\infty$ implies that the following is true for any $\mathbb{Q}\in\mathcal{Q}$:
		\begin{align*}
		\rho^\mu(X)&=\int_{\mathcal{I}}\left( \sup\limits_{\mathbb{Q}\in\mathcal{Q}}\left\lbrace E_{\mathbb{Q}}\left[-X\right]-\alpha^{min}_{\rho^i}(\mathbb{Q}) \right\rbrace\right) d\mu\\
		&\geq\sup\limits_{\left\lbrace \int_{\mathcal{I}}\mathbb{Q}^id\mu=\mathbb{Q},\:i\to\alpha^{min}_{\rho^i}\left(\mathbb{Q}^i\right)\in\mathcal{G}\right\rbrace }\left\lbrace \int_{\mathcal{I}}\left( E_{\mathbb{Q}^i}[-X]-\alpha^{min}_{\rho^i}\left(\mathbb{Q}^i\right) \right)  d\mu \right\rbrace\\
		&\geq\sup\limits_{\mathbb{Q}\in\mathcal{Q}}\left\lbrace E_\mathbb{Q}[-X]-\alpha_{\rho^\mu}(\mathbb{Q}) \right\rbrace.
		\end{align*}
		The last inequality is due to Lemma \ref{Lmm:probs} and the fact that $\alpha^\mu(\mathbb{Q})=\infty$ when $\int_{\mathcal{I}}\mathbb{Q}^id\mu=\mathbb{Q}$ but $i\to\alpha^{min}_{\rho^i}\left(\mathbb{Q}^i\right)$ is not $\mathcal{G}$-measurable. For the converse, consider for each $n\in\mathbb{N}$ the measurable (possibly empty) partition $P^n$ of $\mathcal{I}$ as $P^n=\left\lbrace t^n_k,\:k=1,\dots,n\right\rbrace $. Define  $\alpha_{t^n_k}(\mathbb{Q})=\sup_{i\in t^n_k}\alpha^{\min}_{\rho^i}(\mathbb{Q})$ for any $\mathbb{Q}\in\mathcal{Q}$, with the convention that $\sup\emptyset=0$. Further, for each $n\in\mathbb{N}$ define the map $\alpha^\mu_n\colon\mathcal{Q}\to\mathbb{R}_+\cup\{\infty\}$  as 
		\[\alpha^\mu_n(\mathbb{Q})=\inf\left\lbrace \sum_{k=1}^n\alpha_{t^n_k}(\mathbb{Q}^{t^n_{k}})\mu(t^n_k)\colon \mathbb{Q}=\sum_{k=1}^{n}\mathbb{Q}^{t^n_{k}}\mu(t^n_k),\:\mathbb{Q}^{t^n_{k}}\in\mathcal{Q}\:\forall\:k\in\{1,\dots,n\}\right\rbrace. \] It is clear that $\alpha_n^\mu(\mathbb{Q})\downarrow\alpha^{\rho^\mu}(\mathbb{Q})$ for each $\mathbb{Q}\in\mathcal{Q}$. Define for each $n\in\mathbb{N}$  the map 
		\[\rho^\mu_n(X)=\sup\limits_{\mathbb{Q}\in\mathcal{Q}}\left\lbrace E_\mathbb{Q}[-X]-\alpha^\mu_n(\mathbb{Q})\right\rbrace,\:X\in L^\infty.\] We then have that $\rho^\mu_n(X)\uparrow \rho^\mu(X)$ for any $X\in L^\infty$. Thus, we get for any $X\in L^\infty$ that \begin{align*}
		\rho^\mu(X)&\geq\sup\limits_{\mathbb{Q}\in\mathcal{Q}}\left\lbrace E_\mathbb{Q}[-X]-\alpha_{\rho^\mu}(\mathbb{Q}) \right\rbrace\\
		&\geq \sup_n\sup\limits_{\mathbb{Q}\in\mathcal{Q}}\left\lbrace E_\mathbb{Q}[-X]-\alpha_{\rho^\mu_n}(\mathbb{Q}) \right\rbrace\\
		&=\sup_n \rho_n^\mu(X)=\rho^\mu(X).
		\end{align*}		
		
		Hence, $\rho^\mu(X)=\sup\limits_{\mathbb{Q}\in\mathcal{Q}}\left\lbrace E_\mathbb{Q}[-X]-\alpha_{\rho^\mu}(\mathbb{Q}) \right\rbrace,\:\forall\:X\in L^\infty$.   

		\item 
		We begin by showing that $\mathcal{Q}_\rho^\mu$ satisfies the necessary properties.
		Since every $\mathcal{Q}_{\rho^i}$ is non-empty, we have that $\mathcal{Q}_{\rho^\mu}$ posses at least one element $\mathbb{Q}\in\mathcal{Q}$ such that $\mathbb{Q}=\int_{\mathcal{I}}\mathbb{Q}^id\mu,\mathbb{Q}^i\in\mathcal{Q}_{\rho^i}\:\forall\:i\in\mathcal{I}$. Let $\mathbb{Q}_1,\mathbb{Q}_2\in\mathcal{Q}_{\rho^\mu}$. Then, we have for any $\lambda\in[0,1]$ that $\lambda \mathbb{Q}_1+(1-\lambda)\mathbb{Q}_2=\int_{\mathcal{I}}\left( \lambda \mathbb{Q}_1^i+(1-\lambda)\mathbb{Q}_2^i\right) d\mu$. Since $\mathcal{Q}_{\rho^i}$ is convex for any $i\in\mathcal{I}$ we have that $\lambda \mathbb{Q}_1+(1-\lambda)\mathbb{Q}_2\in\mathcal{Q}_{\rho^\mu}$ as desired. To demonstrate that taking closure does not affect the supremum, let  $\{\mathbb{Q}_n\}_{n=1}^\infty\in\mathcal{Q}_\rho^{\mu}$ such that $\mathbb{Q}_n\rightarrow \mathbb{Q}$ in the total variation norm. Then we have
		\[E_\mathbb{Q}[-X]=\lim\limits_{n\rightarrow\infty}E_{\mathbb{Q}_n}[-X]\leq\sup\limits_{n}E_{\mathbb{Q}_n}[-X]\leq\sup\limits_{\mathbb{Q}\in\mathcal{Q}_{\rho^\mu}}E_\mathbb{Q}[-X].\] In view of  [PH], it is enough to show that $\alpha^{min}_{\rho^\mu}$ is a convex indicator function on $cl\left(\mathcal{\mathbb{Q}}^\mu_\rho\right)$, i.e. it assumes $0$ on $cl\left(\mathcal{\mathbb{Q}}^\mu_\rho\right)$ and $\infty$ otherwise. Note that $\alpha^{min}_{\rho^i}(\mathbb{Q}^i)=0,\:\forall\:\mathbb{Q}^i\in\mathcal{Q}_\rho^i$. Thus, $\alpha_{\rho^\mu}=0$ in $\mathcal{Q}_{\rho^\mu}$ and we have that $0\leq\alpha_{\rho^\mu}^{min}(\mathbb{Q})\leq\alpha_{\rho^\mu}(\mathbb{Q})=0,\:\forall\:\mathbb{Q}\in\mathcal{Q}_{\rho^\mu}$. Due to the lower semi-continuity property, we have that $\alpha^{min}_{\rho^\mu}(\mathbb{Q})=0$ for any limit point $\mathbb{Q}$ of sequences in $\mathcal{Q}_{\rho^\mu}$. Let $\mathbb{Q}\in\mathcal{Q}\backslash cl\left(\mathcal{Q}_{\rho^\mu}\right)$, and assume toward contradiction that  $\alpha^{min}_{\rho^\mu}(\mathbb{Q})=0$. By the Hahn-Banach Theorem, we can find, under some standardization, if needed, a $X\in L^\infty$ such that
		$E_\mathbb{Q}[-X]>\sup_{\mathbb{Q}\in cl\left(\mathcal{Q}_{\rho^\mu}\right)}E_\mathbb{Q}[-X]$. Note that $\alpha^\mu(\mathbb{Q})=\infty$ for any $\mathbb{Q}\in\mathcal{Q}\backslash \mathcal{Q}_{\rho^\mu}$ since $\mu\left(\alpha^i(\mathbb{Q})=\infty \right)>0$. Thus, $\alpha^\mu(\mathbb{Q})$ is a convex indicator function over $\mathcal{Q}_{\rho^\mu}$. We then get,
		\begin{align*}
		\alpha^{\min}_{\rho^\mu}(\mathbb{Q})&\geq E_\mathbb{Q}[-X]-\rho^\mu(X)\\
		&>\sup_{\mathbb{Q}\in cl\left(\mathcal{Q}_{\rho^\mu}\right)}E_\mathbb{Q}[-X]-\rho^\mu(X)\\&\geq \sup_{\mathbb{Q}\in \mathcal{Q}_{\rho^\mu}}E_\mathbb{Q}[-X]-\rho^\mu(X)\\
		&=\sup\limits_{\mathbb{Q}\in\mathcal{Q}}\left\lbrace E_\mathbb{Q}[-X]-\alpha_{\rho^\mu}(\mathbb{Q}) \right\rbrace-\rho^\mu(X)=0.
		\end{align*}
		This deduction is a contradiction. Then, we must have $\alpha^{min}_{\rho^\mu}(\mathbb{Q})=\infty$.
	\end{enumerate}
\end{proof}

\begin{Rmk}
	We have that $\alpha_{\rho^\mu}$, in this case, can be understood as some extension of the concept of inf-convolution for arbitrary terms represented by theoretical and integral concepts. The sum of finite risk measures leads to the inf-convolutions of their penalty functions. By extrapolating the argument, such a result is also useful regarding available conjugates for an arbitrary mixture of convex functionals. 
\end{Rmk}

\begin{Rmk}
	Note that we could consider the families $\{\mathbb{Q}^i\in\mathcal{P},\:i\in\mathcal{I}\}$ that define both $\alpha_{\rho^\mu}$ and $\mathcal{Q}_{\rho^\mu}$ by belonging to determined sets in terms of $\mu$-a.s. instead of point-wise in $\mathcal{I}$. This claim is true because the criterion is Lebesgue integral concerning each specified $\mu\in\mathcal{V}$. We choose the point-wise option in order to keep the pattern since we have not assumed fixed probability on $(\mathcal{I},\mathcal{G})$ alongside the text. Moreover, the integrals that define both $\alpha_{\rho^\mu}$ and $\mathcal{Q}_{\rho^\mu}$ may also be understood in the sense of Bochner
	integral, see \cite{Aliprantis2006} chapter 11 for details.
\end{Rmk}

\begin{Rmk}\label{rmk:lsc}
	We have that $\alpha_{\rho^\mu}$ is convex and lower semi-continuous if and only if it coincides with  $\alpha_{\rho^\mu}^{min}$. This because when $\alpha_{\rho^\mu}$ is convex and lower semi-continuous, by bi-duality regarding Legendre-Fenchel conjugates and the fact that $\rho^\mu=(\alpha_{\rho^\mu})^{*}$, we obtain $\alpha_{\rho^\mu}=(\alpha_{\rho^\mu})^{**}=(\rho^\mu)^{*}=\alpha_{\rho^\mu}^{min}$. Thus, $\alpha_{\rho^\mu}^{min}$ is the lower semi-continuous hull of $\alpha_{\rho^\mu}$ in the sense that we can obtain the first by closing the epigraph of $\alpha_{\rho^\mu}$ in $\mathcal{Q}\times\mathbb{R}_+\cup\{\infty\}$.  In the case of finite cardinality for $\mathcal{I}$, $\mathcal{Q}_{\rho^\mu}$ is closed as exposed in Proposition 2.1 of \cite{Ang2018}, which makes it possible to drop the closure on \eqref{eq:cohmu} in such situation.
\end{Rmk}

We now have the necessary conditions to enunciate the main result in this section, which represents composed risk measures in the usual framework of Theorem \ref{the:dual}.

\begin{Thm}\label{thm:dualcomp}
	Let $\rho_\mathcal{I}=\{\rho^i\colon L^\infty\rightarrow\mathbb{R},\:i\in\mathcal{I}\}$ be a collection of Fatou continuous convex risk measures, $f\colon \mathcal{X}\rightarrow\mathbb{R}$ possessing [M], [TI], [C] and [FC], and  $\rho\colon L^\infty\rightarrow\mathbb{R}$ defined as $\rho(X)=f(\rho_\mathcal{I}(X))$. Then:
	\begin{enumerate}
		\item $\rho$ can be represented as
		\begin{equation}\label{eq:compdual}
		\rho(X)=\sup\limits_{\mathbb{Q}\in\mathcal{Q}}\left\lbrace E_\mathbb{Q}[-X]-\alpha_\rho(\mathbb{Q}) \right\rbrace,\:\forall\:X\in L^\infty, 
		\end{equation} 
		where $\alpha_\rho(\mathbb{Q})=\inf\limits_{\mu\in\mathcal{V}}\left\lbrace\alpha_{\rho^\mu}(\mathbb{Q})+\gamma_f(\mu) \right\rbrace$, with $\gamma_f$ and $\alpha_{\rho^\mu}$ defined as in \eqref{eq:penf} and \eqref{eq:pennu}, respectively.
		\item If in addition to the initial hypotheses $f$ possesses [PH], then the penalty term becomes $\alpha_\rho(\mathbb{Q})=\inf\limits_{\mu\in\mathcal{V}_f}\alpha_{\rho^\mu}(\mathbb{Q})$, where $\mathcal{V}_f$ is as in Remark \ref{Rmk:dualf}.
		\item If in addition to the initial hypotheses $\rho^i$ possess, for any $i\in\mathcal{I}$, [PH], then $\alpha_\rho(\mathbb{Q})=\infty\:\forall\:\mathbb{Q}\in\mathcal{Q}\backslash\cup_{\mu\in\mathcal{V}}cl\left(\mathcal{Q}_{\rho^\mu}\right)$, where $cl$ means closure in the total variation norm.
		\item If, in addition to the initial hypotheses, we have the situations in (ii) and (iii), then the representation of $\rho$ becomes
		\begin{equation}
		\rho(X)=\sup\limits_{\mathbb{Q}\in\mathcal{Q}_\rho^{\mathcal{V}_f}}E_\mathbb{Q}[-X],\:\forall\:X\in L^\infty,
		\end{equation}
		where $\mathcal{Q}_\rho^{\mathcal{V}_f}$ is the closed convex hull of $\cup_{\mu\in\mathcal{V}_f}cl(\mathcal{Q}_{\rho^\mu})$.
	\end{enumerate}
\end{Thm}

\begin{proof}
	From the hypotheses and Proposition \ref{prp:prop}, we have that $\rho$ is a Fatou continuous convex risk measure.
	\begin{enumerate}
		\item From Lemma \ref{lmm:f} and Theorem \ref{lmm:dualmu} we have that
		\begin{align*}
		\rho(X)&=\sup\limits_{\mu\in\mathcal{V}}\left\lbrace\sup\limits_{\mathbb{Q}\in\mathcal{Q}}\left[ E_\mathbb{Q}[-X]-\alpha_{\rho^\mu}(\mathbb{Q}) \right]-\gamma_f(\mu) \right\rbrace\\
		&=\sup\limits_{\mathbb{Q}\in\mathcal{Q}}\left\lbrace E_\mathbb{Q}[-X]-\inf\limits_{\mu\in\mathcal{V}}\left[\alpha_{\rho^\mu}(\mathbb{Q}) +\gamma_f(\mu)\right] \right\rbrace\\
		&=\sup\limits_{\mathbb{Q}\in\mathcal{Q}}\left\lbrace E_\mathbb{Q}[-X]-\alpha_\rho(\mathbb{Q}) \right\rbrace.
		\end{align*}
		
		\item If $f$ possesses [PH], then $\gamma_f$ assumes value $0$ in $\mathcal{V}_f$ and $\infty$ otherwise. Thus, we get \[\alpha_\rho(\mathbb{Q})=\inf\limits_{\mu\in\mathcal{V}}\left\lbrace\alpha_{\rho^\mu}(\mathbb{Q})+\gamma_f(\mu) \right\rbrace=\inf\limits_{\mu\in\mathcal{V}_f}\alpha_{\rho^\mu}(\mathbb{Q}).\]
		\item When each element of $\rho_\mathcal{I}$ fulfills [PH], we have that $\alpha_{\rho^\mu}(\mathbb{Q})=\infty\:\forall\:\mu\in\mathcal{V}$ for any $\mathbb{Q}\in\mathcal{Q}\backslash\cup_{\mu\in\mathcal{V}}cl(\mathcal{Q}_{\rho^\mu})$. One gets the claim by adding the non-negative term $\gamma_f(\mu)$ and taking the infimum over $\mathcal{V}$.
		\item In this context, the generated $\rho$ is coherent from Theorem \ref{prp:prop}. Moreover, in this case from Lemma \ref{lmm:f} and Proposition \ref{lmm:dualmu} together to items (ii) and (iii) we have that
		\begin{align*}
		\rho(X)&=\sup\limits_{\mu\in\mathcal{V}_f}\sup\limits_{\mathbb{Q}\in cl(\mathcal{Q}_{\rho^\mu})}E_\mathbb{Q}[-X]\\
		&=\sup\limits_{\mathbb{Q}\in\cup_{\mu\in\mathcal{V}_f}cl(\mathcal{Q}_{\rho^\mu})}E_\mathbb{Q}[-X]\\
		&=\sup\limits_{\mathbb{Q}\in\mathcal{Q}_\rho^{\mathcal{V}_f}}E_\mathbb{Q}[-X].
		\end{align*}
		In order to verify that the supremum is not altered by considering the closed convex hull, let $\mathbb{Q}_1,\mathbb{Q}_2\in\cup_{\mu\in\mathcal{V}_f}cl(\mathcal{Q}_{\rho^\mu})$ and $\mathbb{Q}=\lambda \mathbb{Q}_1+(1-\lambda)\mathbb{Q}_2,\:\lambda\in[0,1]$. Then \[E_{\mathbb{Q}}[-X]\leq\max(E_{\mathbb{Q}_1}[-X],E_{\mathbb{Q}_2}[-X])\leq\sup\limits_{\mathbb{Q}\in\cup_{\mu\in\mathcal{V}_f}cl(\mathcal{Q}_{\rho^\mu})}E_\mathbb{Q}[-X],\]
		thus convex combinations do not alter the supremum. For closure, the deduction is quite similar to that used in the proof of Theorem \ref{lmm:dualmu}. 
	\end{enumerate}
\end{proof}

\begin{Rmk}
	Note that when $f=f^\mu$, we recover the result in Theorem \ref{lmm:dualmu}. Moreover,  $\mathcal{Q}_\rho^{\mathcal{V}_f}\subseteq\mathcal{Q}_{\rho^{WC}}$ since $f\leq f^{WC}$ for any bounded combination $f$. Furthermore, when $\alpha_{\rho^\mu}$ is convex and lower semi-continuous, we have that $\alpha_{\rho}$ coincides with the minimal penalty term because
	\begin{align*}
	\alpha_\rho^{min}(\mathbb{Q})&=\sup\limits_{X\in L^\infty}\left\lbrace E_\mathbb{Q}[-X]-f(R_X) \right\rbrace \\
	&=\sup\limits_{X\in L^\infty}\left\lbrace E_\mathbb{Q}[-X]-\sup\limits_{\mu\in\mathcal{V}}\left\lbrace\rho^{\mu}(X)-\gamma_f(\mu) \right\rbrace  \right\rbrace\\
	&=\inf\limits_{\mu\in\mathcal{V}}\left\lbrace \gamma_f(\mu)+\sup\limits_{X\in L^\infty}\left\lbrace E_\mathbb{Q}[-X]-\rho^\mu(X) \right\rbrace \right\rbrace 	\\
	&=\alpha_{\rho}(\mathbb{Q}).
	\end{align*}
	Hence, the reasoning in Remark \ref{rmk:lsc} is also valid in here.
\end{Rmk}

\begin{Rmk}
	Under [CS] or [R], the supremum over $\mathcal{Q}$ can be replaced, respectively, by sub-probabilities (measures on $(\Omega,\mathcal{F})$ with $\mathbb{Q}(\Omega)\leq 1$) or probabilities equivalent to $\mathbb{P}$.  Under [QC] for $f$ and dropping its [C] and [TI] we get the representation $\rho(X)=\sup_{\mu\in\mathcal{V}}R_f(\rho^\mu(X),\mu)$, where $R_f\colon\mathbb{R}\times\mathcal{V}\rightarrow\mathbb{R}$ is defined as $R_f(x,\mu)=\inf\left\lbrace f(R)\colon \int_{\mathcal{I}}Rd\mu=x\right\rbrace $. See the papers in Remark \ref{rmk:prop} for details.
\end{Rmk}

Regarding the specific case of $\rho^{WC}$, Proposition 9 in \cite{Follmer2002} states that it can be represented by, the non necessarily convex, $\alpha_{\rho^{WC}}(\mathbb{Q})=\inf\limits_{i\in\mathcal{I}}\alpha^{min}_{\rho^i}(\mathbb{Q}),\:\forall\:\mathbb{Q}\in\mathcal{Q}$. Under coherence, Theorem 2.1 of \cite{Ang2018} claims that $\mathcal{Q}_{\rho^{WC}}=conv(\cup_{i=i}^n\mathcal{Q}_{\rho^i})$ when $\mathcal{I}$ is finite with cardinality $n$. We now expose a result that states these facts under our approach, which is more general.

\begin{Prp}\label{prp:dualWC2}
	Let $\{\rho^i: L^\infty\rightarrow \mathbb{R},\:i\in\mathcal{I}\}$ be a collection of Fatou continuous convex risk measures, and  $\rho^{WC}\colon L^\infty\rightarrow\mathbb{R}$ defined as  in \eqref{Ex:WC}. Then:
	\begin{enumerate}
		\item  $\alpha_{\rho^{WC}}(\mathbb{Q})=\inf\limits_{\mu\in\mathcal{V}}\alpha_{\rho^\mu}(\mathbb{Q})=\inf\limits_{i\in\mathcal{I}}\alpha^{min}_{\rho^i}(\mathbb{Q}),\:\forall\:\mathbb{Q}\in\mathcal{Q}$.
		\item  If in addition to the initial hypotheses $\rho^i$ possess, for any $i\in\mathcal{I}$, [PH], then $\mathcal{Q}_{\rho^{WC}}=\mathcal{Q}_{\rho}^{\mathcal{V}}$, which is the closed convex hull of $\cup_{i\in\mathcal{I}}\mathcal{Q}_{\rho^i}$. 
	\end{enumerate}	
	
\end{Prp}

\begin{proof}
	From Propositions \ref{prp:fmu} and \ref{prp:prop}, we have that $\rho^{WC}$ is a Fatou continuous convex risk measure when all $\rho^i$ also are. Moreover, from the fact that $|\rho^{WC}(X)|\leq\lVert X\rVert_\infty<\infty$  we have that $\rho^{WC}$ takes only finite values.
	\begin{enumerate}
		\item For fixed $\mathbb{Q}\in\mathcal{P}$, we have that for any $\epsilon>0$, there is $j\in\mathcal{I}$ such that \[\inf\limits_{i\in\mathcal{I}}\alpha^{min}_{\rho^i}(\mathbb{Q})\leq\alpha_{\rho^j}^{min}(\mathbb{Q})\leq \inf\limits_{i\in\mathcal{I}}\alpha^{min}_{\rho^i}(\mathbb{Q})+\epsilon.\]	Recall that $\inf\limits_{i\in\mathcal{I}}\alpha^{min}_{\rho^i}(\mathbb{Q})\leq\alpha_{\rho^\mu}(\mathbb{Q})\leq\int_{\mathcal{I}}\alpha_{\rho^i}^{min}(\mathbb{Q})d\mu$ for any $\mu\in\mathcal{V}$. Then it is true that for any $\epsilon>0$, there is $\mu\in\mathcal{V}$ such that \[\inf\limits_{i\in\mathcal{I}}\alpha^{min}_{\rho^i}(\mathbb{Q})\leq\alpha_{\rho^\mu}(\mathbb{Q})\leq \inf\limits_{i\in\mathcal{I}}\alpha^{min}_{\rho^i}(\mathbb{Q})+\epsilon.\]	
		By taking the infimum over $\mathcal{V}$ and since $\epsilon$ was taken arbitrarily, we get that $\alpha_{\rho^{WC}}(\mathbb{Q})=\inf\limits_{i\in\mathcal{I}}\alpha^{min}_{\rho^i}(\mathbb{Q})$.
		\item 	From Propositions \ref{prp:fmu} and \ref{prp:prop} we have that $\rho^{WC}$ is a Fatou continuous coherent risk measure when all $\rho^i$ also are. Thus, in light of Theorem \ref{the:dual}, it has a dual representation. We then have
		\[\rho^{WC}(X)=\sup\limits_{i\in\mathcal{I}}\sup\limits_{\mathbb{Q}\in\mathcal{Q}_{\rho^i}}E_\mathbb{Q}\left[-X\right]=\sup\limits_{\mathbb{Q}\in\cup_{i\in\mathcal{I}}\mathcal{Q}_{\rho^i}}E_\mathbb{Q}[-X].\]
		The fact that supremum is not altered by considering the closed convex hull follows similar steps as those in the proof of Theorem \ref{thm:dualcomp}. We have that $\cup_{i\in\mathcal{I}}\mathcal{Q}_{\rho^i}$ is non-empty because every $\mathcal{Q}_{\rho^i}$ contains at least one element. 
		Hence,  $\mathcal{Q}_{\rho^{WC}}$ coincides to the closed convex hull of $\cup_{i\in\mathcal{I}}\mathcal{Q}_{\rho^i}$. Regarding the equivalence with $\mathcal{Q}_{\rho}^{\mathcal{V}}$, note that for any $i\in\mathcal{I}$ we have that $\mathbb{Q}\in\mathcal{Q}_{\rho^i}$ if and only if $\mathbb{Q}\in cl\mathcal{Q}_{\rho^{\delta_i}}$, where $\delta_i\in\mathcal{V}$ is defined as $\delta_i(A)=1_{A}(i),\:\forall\:A\in\mathcal{G}$. Thus, we get that $\cup_{i\in\mathcal{I}}\mathcal{Q}_{\rho^i}\subseteq\cup_{\mu\in\mathcal{V}}cl(\mathcal{Q}_{\rho^\mu})$. By considering closed convex hulls we obtain $\mathcal{Q}_{\rho^{WC}}\subseteq\mathcal{Q}_{\rho}^{\mathcal{V}}$. For the converse relation note that if $\mathbb{Q}\in \cup_{\mu\in\mathcal{V}}cl(\mathcal{Q}_{\rho^\mu})$, then $\alpha_{\rho^{WC}}^{min}(\mathbb{Q})\leq\inf\limits_{\mu\in\mathcal{V}}\alpha_{\rho^\mu}(\mathbb{Q})=\inf\limits_{i\in\mathcal{I}}\alpha^{min}_{\rho^i}(\mathbb{Q})=0$. Hence, $\mathbb{Q}\in \mathcal{Q}_{\rho^{WC}}$ as desired.
	\end{enumerate}	
\end{proof}

\subsection{Law invariant case}

Under [LI] of the components in $\rho_\mathcal{I}$, the generated $\rho$ is representable in light of those formulations in Theorem \ref{thm:dual2}. We begin with an auxiliary result for the representation when $\rho^\mu$ is law invariant. The next Proposition follows in this direction.

\begin{Prp}\label{lmm:dualmu2}
	Let $\rho_\mathcal{I}=\{\rho^i\colon L^\infty\rightarrow\mathbb{R},\:i\in\mathcal{I}\}$ be a collection of law invariant convex risk measures and $\rho^{\mu}: L^\infty\rightarrow \mathbb{R}$ defined as in \eqref{eq:mu}. Then:
	\begin{enumerate}
		\item  $\rho^\mu$ can be represented as:
		\begin{equation}
		\rho^\mu(X)=\sup\limits_{m\in\mathcal{M}}\left\lbrace \int_{(0,1]}ES^\alpha(X)dm-\beta_{\rho^\mu}(m)\right\rbrace,\:\forall\:X\in L^\infty,
		\end{equation}
		with convex  $\beta_{\rho^\mu}\colon\mathcal{M}\rightarrow\mathbb{R}_+\cup\{\infty\}$, defined as 
		\begin{equation}\label{eq:pennu2}
		\beta_{\rho^\mu}(m)=\inf\left\lbrace \begin{cases*}\int_\mathcal{I}\beta_{\rho^i}^{min}(m^i)d\mu,\text{if}\:i\to\beta^{min}_{\rho^i}\left(m^i\right)\in\mathcal{G}
		\\
		\infty,\:\text{otherwise}\end{cases*} 
		\colon \int_{\mathcal{I}}m^id\mu=m,\:m^i\in\mathcal{M}\:\forall\:i\in\mathcal{I}\right\rbrace .
		\end{equation}

		\item	If in addition $\rho^i$ fulfills, for every $i\in\mathcal{I}$, [PH], then the representation is 
		\begin{equation}\label{eq:cohnu2}
		\rho^\mu(X)=\sup\limits_{m\in cl(\mathcal{M}_{\rho^\mu})} \int_{(0,1]}ES^\alpha(X)dm,\:\forall\:X\in L^\infty,
		\end{equation}
		with $\mathcal{M}_{\rho^\mu}=\left\lbrace m\in\mathcal{M}\colon m=\int_{\mathcal{I}}m^id\mu,m^i\in\mathcal{M}_{\rho^i}\:\forall\:i\in\mathcal{I}\right\rbrace$ non-empty and convex, where $cl$ means the closure in total variation norm.
		\item  If $\rho^i$ also is, for every $i\in\mathcal{I}$, comonotone, then the representation is 
		\begin{equation}\label{eq:rhomucom}
		\rho^\mu(X)= \int_{(0,1]}ES^\alpha(X)dm,\:\forall\:X\in L^\infty,
		\end{equation}
		where $m\in cl(\mathcal{M}_{\rho^\mu})$.	
	\end{enumerate}
\end{Prp}

\begin{proof}
	From the hypotheses and Proposition \ref{prp:prop} together to Theorem \ref{thm:fatou}, we have that $\rho^\mu$ is a law invariant convex risk measure. Theorem \ref{thm:fatou} assures its [FC]. 	From Theorems \ref{the:dual} and \ref{thm:dual2}, we have for any $m\in\mathcal{M}$ there is $\mathbb{Q}^\prime$ such that \[\beta_{\rho^i}^{min}(m)=\sup\left\lbrace\alpha_{\rho^i}^{min}(\mathbb{Q})\colon \frac{d\mathbb{Q}}{d\mathbb{P}}\sim \frac{d\mathbb{Q}^\prime}{d\mathbb{P}},\int_{(u,1]}\frac{1}{v}dm
	=F^{-1}_{\frac{d\mathbb{Q}^\prime}{d\mathbb{P}}}(1-u),\mathbb{Q}^\prime\in\mathcal{Q} \right\rbrace=\alpha_{\rho^i}^{min}(\mathbb{Q}^\prime).\] Hence, from Assumption \ref{Lmm:measure} the maps $i\rightarrow\beta^{min}_{\rho^{i}}(m)=\alpha_{\rho^i}^{min}(\mathbb{Q}^\prime)$ are $\mathcal{G}$-measurable for any $m\in\mathcal{M}$. From that, the proof follows similar steps to those of Theorem \ref{lmm:dualmu} with $m^i\rightarrow\int_{(0,1]}ES^\alpha(X)dm^i$ linear and playing the role of $\mathbb{Q}^i\rightarrow E_{\mathbb{Q}^i}[-X]$. For the comonotonic case in (iii), the result is due to the supremum in \eqref{eq:dual2com} being attained for each $\rho^i$.	
\end{proof}

\begin{Rmk}\label{rmk:spectralmu}
	The representation in item (iii) on \eqref{eq:rhomucom} is equivalent to the spectral one as \begin{equation}\label{eq:spectralmu}
	\rho^\mu(X)=\int_{0}^{1}VaR^\alpha(X)\phi^\mu(\alpha)d\alpha,
	\end{equation} where $\phi^\mu(\alpha)=\int_{\mathcal{I}}\phi^i(\alpha)d\mu$, and $\phi^i\colon[0,1]\rightarrow[0,1]$ is as in Example \ref{ex:spec} for any $i\in\mathcal{I}$. The map $i\rightarrow\phi^i(\alpha)$ is $\mathcal{G}$-measurable to any $\alpha\in[0,1]$. To verify this claim, note that for each $i\in\mathcal{I}$ it is true that $\phi^i(\alpha)=\int_{(\alpha,1]}\frac{1}{s}dm^i(s),\:m^i\in\mathcal{M}$. Then $\frac{d\nu^i}{dm^i}=\frac{1}{s}$ defines a finite measure on $(0,1]$.  Thus, from definition of $\mathcal{M}_{\rho^\mu}$, $i\rightarrow\nu^i(\alpha,1]=\phi^i(\alpha)$ is $\mathcal{G}$-measurable for any $\alpha\in[0,1]$.
\end{Rmk}

We can now propose a result for the dual representation under [LI].	The next Corollary exposes such content.

\begin{Crl}\label{cor:dualcomp2}
	Let $\rho_\mathcal{I}=\{\rho^i\colon L^\infty\rightarrow\mathbb{R},\:i\in\mathcal{I}\}$ be a collection of law invariant convex risk measures, $f\colon \mathcal{X}\rightarrow\mathbb{R}$ possessing [M], [TI], [C] and [FC], and  $\rho\colon L^\infty\rightarrow\mathbb{R}$ defined as $\rho(X)=f(\rho_\mathcal{I}(X))$. Then:
	\begin{enumerate}
		\item $\rho$ can be represented as
		\begin{equation}\label{eq:compdual2}
		\rho(X)=\sup\limits_{m\in\mathcal{M}}\left\lbrace \int_{(0,1]}ES^\alpha(X)dm-\beta_\rho(m) \right\rbrace,\:\forall\:X\in L^\infty, 
		\end{equation} 
		where $\beta_\rho(m)=\inf\limits_{\mu\in\mathcal{V}}\left\lbrace\beta_{\rho^\mu}(m)+\gamma_f(\mu) \right\rbrace$, with $\gamma_f$ and $\beta_{\rho^\mu}$ defined as in \eqref{eq:penf} and \eqref{eq:pennu2}, respectively.
		\item If in addition to the initial hypotheses $f$ possesses [PH], then the penalty term becomes $\beta_\rho(m)=\inf\limits_{\mu\in\mathcal{V}_f}\beta_{\rho^\mu}(m)$, where $\mathcal{V}_f$ is as in Remark \ref{Rmk:dualf}.
		\item If in addition to the initial hypotheses $\rho^i$ possess, for any $i\in\mathcal{I}$, [PH], then $\beta_\rho(m)=\infty,\:\forall\:m\in\mathcal{M}\backslash\cup_{\mu\in\mathcal{V}}cl(\mathcal{M}_{\rho^\mu})$, where $cl$ means closure in the total variation norm.
		\item If, in addition to the initial hypotheses, we have the situations in (ii) and (iii), then the representation of $\rho$ becomes
		\begin{equation}\label{eqcomdual2b}
		\rho(X)=\sup\limits_{m\in\mathcal{M}_\rho^{\mathcal{V}_f}}\int_{(0,1]}ES^\alpha(X)dm,\:\forall\:X\in L^\infty,
		\end{equation}
		where $\mathcal{M}_\rho^{\mathcal{V}_f}$ is the closed convex hull of $\cup_{\mu\in\mathcal{V}_f}cl(\mathcal{M}_{\rho^\mu})$.
		\item If in addition to the initial hypotheses $\rho^i$ possess, for any $i\in\mathcal{I}$, [CA], then $\beta_\rho(m)=\infty,\:\forall\:m\in\mathcal{M}\backslash\cup_{\mu\in\mathcal{V}}\{m^\mu_c\}$, where \[m^\mu_c=\argmax\limits_{m\in cl(\mathcal{M}_{\rho^\mu})}\int_{(0,1]}ES^\alpha(X)dm,\:\forall\:\mu\in\mathcal{V}.
		\]
		\item If, in addition to the initial hypotheses, we have (ii) and (v), then the representation of $\rho$ becomes
		\begin{equation}\label{eqcomdual2c}
		\rho(X)=\sup\limits_{m\in\mathcal{M}_{\rho,c}^{\mathcal{V}_f}}\int_{(0,1]}ES^\alpha(X)dm,\:\forall\:X\in L^\infty,
		\end{equation}
		where $\mathcal{M}_{\rho,c}^{\mathcal{V}_f}$ is the closed convex hull of $\cup_{\mu\in\mathcal{V}_f}\{m^\mu_c\}$.
	\end{enumerate}
\end{Crl}

\begin{proof}
	Direct from Theorem \ref{thm:dualcomp} and Proposition \ref{lmm:dualmu2}.
\end{proof}

\begin{Rmk}
	Since the comonotonicity of a pair $X,Y$ does not imply the same property for the pair $R_X,R_Y$, the only situation where $\rho$ is surely comonotone occurs, from Proposition \ref{prp:prop} and Lemma \ref{lmm:f}, when $\{f^\mu\}_{\mu\in\mathcal{V}}$. In this case, we have
	\begin{equation}
	\rho(X)=\rho^\mu(X)=\int_{(0,1]}ES^\alpha(X)dm^\mu_c,\:\forall\:X\in L^\infty.
	\end{equation}
	From \eqref{eq:spectralmu}, we have that $\phi^\mu(\alpha)=\int_{(\alpha,1]}\frac{1}{s}dm^\mu_c(s)$, where $\phi^\mu$ is as in Remark \ref{rmk:spectralmu}. 
	It would be possible also to investigate a situation where $f$ possesses representation in terms of ES over $\mathcal{X}$ as in Example \ref{ex:ut}, but we do not consider any base probability on $(\mathcal{I},\mathcal{G})$. However, note that those would be special cases of our framework. 
\end{Rmk}	

\section{Acceptance sets}\label{sec:accept}

\subsection{Properties}

In this section, we expose results regarding the acceptance sets of composed risk measures $\rho=f(\rho_{\mathcal{I}})$ based on the properties of both $\rho_{\mathcal{I}}$ and $f$. Of course, when $f$ possesses [M] we have \[\mathcal{A}_{\rho_{\mathcal{I}}}\supseteq\left\lbrace X\in L^\infty\colon \exists\:R\in f^{-1}(0)\:\text{s.t.}\:\rho_{\mathcal{I}}(X)\leq R\right\rbrace,\] where the inequality is in the point-wise order of $\mathcal{X}$. If in addition $f$ is injective function, then from normalization, $f(0)=0$, we have \[\left\lbrace X\in L^\infty\colon \exists\:R\in f^{-1}(0)\:\text{s.t.}\:\rho_{\mathcal{I}}(X)\leq R\right\rbrace=\left\lbrace X\in L^\infty\colon\rho_{\mathcal{I}}(X)\leq 0\right\rbrace=\mathcal{A}_{\rho^{WC}}.\]
However, since the point-wise order in $\mathcal{X}$ is not total, the set  $\mathcal{A}_{\rho_{\mathcal{I}}}$ can be much larger than the positions that lead to elements in the non-positive cone of $\mathcal{X}$. Thus, to provide a general characterization for $\mathcal{A}_{\rho_{\mathcal{I}}}$ is not trivial.

We begin by translating the role of financial properties preservation from section \ref{sec:prop} for acceptance sets.

\begin{Crl}\label{prp:prop2}
	Let $\rho_\mathcal{I}=\{\rho^i\colon L^\infty\rightarrow\mathbb{R},\:i\in\mathcal{I}\}$ be a collection of risk measures, $f\colon\mathcal{X}\rightarrow\mathbb{R}$, and $\rho\colon L^\infty\rightarrow\mathbb{R}$ a risk measure defined as $\rho(X)=f(\rho_\mathcal{I}(X))$. Then:
	\begin{enumerate}
		\item If $\rho_\mathcal{I}$ is composed of risk measures with [M] and $f$ possesses this same property, then $\mathcal{A}_\rho$ is monotone, i.e. $X\in\mathcal{A}_\rho$, $Y\in L^\infty$ and $Y\geq X$ implies in $Y\in\mathcal{A}_\rho$. In particular, $L^\infty_+\subseteq\mathcal{A}_\rho$.
		\item If $\rho_\mathcal{I}$ is composed of risk measures with [TI] and $f$ possesses this same property, then $\rho(X)=\inf\left\lbrace m\in\mathbb{R}\colon X+m\in\mathcal{A}_\rho\right\rbrace$.
		\item If the conditions in items (i) and (ii) are fulfilled, then $A_\rho$ is non-empty, closed with respect to the supremum norm, $\mathcal{A}_\rho\cap\{X\in L^\infty\colon X<0\}=\emptyset$, and $\inf\{m\in\mathbb{R}\colon m\in\mathcal{A}_\rho\}>-\infty$.
		\item If $\rho_\mathcal{I}$ is composed of risk measures with [C] and $f$ possesses this same property in pair with [M], then  $\mathcal{A}_\rho$ is a convex set.
		\item If $\rho_\mathcal{I}$ is composed of risk measures with [PH] and $f$ possesses this same property, then $\mathcal{A}_\rho$ is a cone.
		\item If $\rho_\mathcal{I}$ is composed of law invariant risk measures, then  $\mathcal{A}_\rho$ is law invariant in the sense of $X\in\mathcal{A}_\rho$ and $X\sim Y$ imply $Y\in\mathcal{A}_\rho$.
		\item If $\rho_\mathcal{I}$ is composed of comonotone risk measures and $f$ fulfills [A], then $\mathcal{A}_\rho$ is stable for sums of comonotonic pairs of random variables.
		\item If the conditions in items (i), (ii), and (iv) are fulfilled, $\rho_\mathcal{I}$ is composed of Fatou continuous risk measures and $f$ has [FC] and [M], then $\mathcal{A}_\rho$ is weak* closed.
	\end{enumerate}
\end{Crl}

\begin{proof}
	The claims are directly obtained by noticing that they are implications from Theorems \ref{Thm:accept} and \ref{the:dual} together to Proposition \ref{prp:prop}.
\end{proof}

\begin{Crl}\label{prp:prop4}
	Let $\rho_\mathcal{I}=\{\rho^i\colon L^\infty\rightarrow\mathbb{R},\:i\in\mathcal{I}\}$ be a collection of risk measures, $f\colon\mathcal{X}\rightarrow\mathbb{R}$, and $\rho\colon L^\infty\rightarrow\mathbb{R}$ a risk measure defined as $\rho(X)=f(\rho_\mathcal{I}(X))$. Then:
	\begin{enumerate}
		\item If $\rho_\mathcal{I}$ is composed of risk measures with [C] and $f$ possesses [M] and [QC], then $\mathcal{A}_\rho$ is a convex set.
		\item If  $\rho_\mathcal{I}$ is composed of risk measures with [M] and [CS] and $f$ possesses [M] and [TI], then \[\begin{cases*}
		\rho(X)\geq\inf\left\lbrace m\in\mathbb{R}\colon X+m\in\mathcal{A}_\rho\right\rbrace ,\:\text{if}\:\rho(X)\leq 0,\\
		\rho(X)\leq\inf\left\lbrace m\in\mathbb{R}\colon X+m\in\mathcal{A}_\rho\right\rbrace ,\:\text{if}\:\rho(X)\geq 0.	
		\end{cases*}\]
		\item  If $\rho_\mathcal{I}$ is composed of risk measures with [R] and $f$ has strict [M], then $\mathcal{A}_\rho\cap\{X\in L^\infty_-\colon\mathbb{P}(X<0)>0\}=\emptyset$.
		\item If $\rho_\mathcal{I}$ is composed of risk measures with [SI] and $f$ has [M] together to $f\geq f^{WC}$, then $X\in \mathcal{A}_\rho$ and $Y^-\leq X^-$ imply $Y\in\mathcal{A}_\rho,\:\forall\:X,Y\in L^\infty$. 
	\end{enumerate}
\end{Crl}

\begin{proof}
	Items (i), (iii) and (iv) are direct consequences from Theorem \ref{Thm:accept} and Proposition \ref{prp:prop3}. For item (ii), Proposition \ref{prp:prop3} implies $\rho$ has  [CS].  Note that it may be restated as $\rho(X-C)\leq\rho(X)+C,\:\forall\:C\in\mathbb{R}_+,\:\forall\:X\in L^\infty$ or $m\rightarrow\rho(X+m)+m$ be non-decreasing in $\mathbb{R}_+$ for any $X\in L^\infty$. With [M], Proposition 2.1 in \cite{Cerreia2011} assures $\rho$ is Lipschitz continuous. Fix $X\in L^\infty$. If $\rho(X)\leq0$, then $\rho(X+\rho(X))=\rho(X-(-\rho(X)))\leq\rho(X)-\rho(X)=0$. Thus, $X+\rho(X)\in\mathcal{A}_\rho$ and $\rho(X)\geq\inf\left\lbrace m\in\mathbb{R}\colon X+m\in\mathcal{A}_\rho\right\rbrace$. If $\rho(X)\geq0$, let $k=\inf\left\lbrace m\in\mathbb{R}\colon X+m\in\mathcal{A}_\rho\right\rbrace$. Thus, $k\geq 0$. For any $m\in\mathbb{R}$ with $X+m\in\mathcal{A}_\rho$, we obtain $\rho(X+k)+k\leq\rho(X+m)+m\leq m$. Then, it is true that $\rho(X+k)+k\leq\inf\left\lbrace m\in\mathbb{R}\colon X+m\in\mathcal{A}_\rho\right\rbrace=k$. Thus, $\rho(X+k)\leq 0$. Hence, $k\geq\rho(X+k)+k\geq\rho(X)$.
\end{proof}
Characterization of the acceptance sets can be made explicit for particular cases. 

\begin{Exm}
	We get the following examples for $\mathcal{A}_{f(\rho_{\mathcal{I}})}$:
	\begin{enumerate}
		\item For $f(R)=\sup_{i\in\mathcal{I}}R(i)$ we obtain $f(\rho_{\mathcal{I}})=\rho^{WC}$. In this case we get \[\mathcal{A}_{\rho^{WC}}=\{X\in L^\infty\colon \rho^i(X)\leq 0\:\forall\:i\in\mathcal{I}\}=\bigcap_{i\in\mathcal{I}}\mathcal{A}_{\rho^i}.\] 
		\item For $f(R)=\int_{\mathcal{I}}Rd\mu$ we obtain $f(\rho_{\mathcal{I}})=\rho^{\mu}$. We then have that \[\mathcal{A}_{\rho^\mu}=\left\lbrace X\in L^\infty\colon\int_{\{\rho^i(X)\leq 0\}}\rho^i(X)d\mu\leq-\int_{\{\rho^i(X)> 0\}}\rho^i(X)d\mu\right\rbrace .\] Note that from Assumption \ref{asp:measure} both $\{\rho^i(X)\leq 0\}$ and $\{\rho^i(X)> 0\}$ are in $\mathcal{G}$ for any $X\in L^\infty$.
		\item 	We have that any spectral (distortion) risk measure 
		$\rho^\phi(X)=\int_0^1VaR^\alpha(X)\phi(\alpha)d\alpha$ is a special case of $\rho^\mu$ by choosing $\rho^i(X)=VaR^i(X)$ and $\mu\ll\lambda$ with $\phi(i)=F^{-1}_{\frac{d\mu}{d\lambda}}(1-i)$. Since both $\alpha\rightarrow VaR^\alpha$ and $\alpha\rightarrow \phi(\alpha)$ are non-increasing, we can pick $\alpha_X\in[0,1]$ dependent of $X\in L^\infty$ such that $VaR^\alpha(X)\phi(\alpha)\geq0$ for any $\alpha<\alpha_X$ and $VaR^\alpha(X)\phi(\alpha)\leq0$ for any $\alpha>\alpha_X$. In this case we get \[\mathcal{A}_{\rho^\phi}=\left\lbrace X\in L^\infty\colon\int_{\alpha_X}^1VaR^\alpha\phi(\alpha)d\alpha\leq-\int^{\alpha_X}_0VaR^\alpha\phi(\alpha)d\alpha\right\rbrace .\] From the properties of VaR in Example \ref{Exm:meas}, Proposition \ref{prp:prop} and Corollary \ref{prp:prop2} this set is norm closed,  monotone, law invariant, a cone and stable for addition of comonotone pairs. If we also have that $\phi$ is non-increasing, the acceptance set is convex and weak* closed.
	\end{enumerate}
\end{Exm}

Nonetheless, direct general characterization of $\mathcal{A}_{\rho_{\mathcal{I}}}$ is not so easy from the complexity that arises from the combination. In the next subsection, we provide a general characterization for the case of convex risk measures.

\subsection{General result}

We now explore a more informative characterization for the acceptance sets of $f(\rho_{\mathcal{I}})$ for the case of convex risk measures from section \ref{sec:main}. In this sense, the next Theorem explores the role of $\mathcal{A}_{\rho^\mu}$ in such a framework.

\begin{Thm}\label{Thm:accept2}
	Let $\rho_\mathcal{I}=\{\rho^i\colon L^\infty\rightarrow\mathbb{R},\:i\in\mathcal{I}\}$ be a collection of Fatou continuous convex risk measures, $f\colon \mathcal{X}\rightarrow\mathbb{R}$ possessing [M], [TI], [C] and [FC], and  $\rho\colon L^\infty\rightarrow\mathbb{R}$ defined as $\rho(X)=f(\rho_\mathcal{I}(X))$. Then:
	\begin{enumerate}
		\item The acceptance set of $\rho$ is given by \begin{equation}\label{eq:accept}
		\mathcal{A}_\rho=\bigcap_{\mu\in\mathcal{V}}\{\mathcal{A}_{\rho^\mu}-\gamma_f(\mu)\}.
		\end{equation}
		\item if in addition to initial hypotheses $f$ fulfills [PH], then the acceptance set of $\rho$ is given by \begin{equation}\label{eq:accept2}
		\mathcal{A}_\rho=\bigcap_{\mu\in\mathcal{V}_f}\mathcal{A}_{\rho^\mu}.
		\end{equation}
	\end{enumerate}
\end{Thm}

\begin{proof}
	\begin{enumerate}
		\item We recall that the acceptance set of any Fatou continuous convex risk measure $\rho$ can be obtained through its penalty term as \begin{align*}
		\mathcal{A}_\rho&=\left\lbrace X\in L^\infty\colon \sup\limits_{ \mathbb{Q}\in\mathcal{Q}}\left\lbrace E_\mathbb{Q}[-X]-\alpha^{min}(\mathbb{Q})\right\rbrace \leq 0\right\rbrace \\
		&=\left\lbrace X\in L^\infty\colon E_\mathbb{Q}[-X]\leq\alpha^{min}(\mathbb{Q})\:\forall\:\mathbb{Q}\in\mathcal{Q}\right\rbrace .
		\end{align*}
		Note that this is equivalent to \[\mathcal{A}_\rho=\left\lbrace X\in L^\infty\colon E_\mathbb{Q}[-X]\leq\alpha(\mathbb{Q})\:\forall\:\mathbb{Q}\in\mathcal{Q}\right\rbrace\] for any, not necessarily minimal, penalty term $\alpha_\rho$ that represents $\rho$. Thus, from Theorem \ref{thm:dualcomp} we obtain
		\begin{align*}
		\mathcal{A}_\rho&=\left\lbrace X\in L^\infty\colon E_\mathbb{Q}[-X]\leq\inf\limits_{\mu\in\mathcal{V}}\left\lbrace \alpha_{\rho^\mu}(\mathbb{Q})+\gamma_f(\mu)\right\rbrace \:\forall\:\mathbb{Q}\in\mathcal{Q}\right\rbrace\\
		&=\left\lbrace X\in L^\infty\colon E_\mathbb{Q}[-X]\leq \alpha_{\rho^\mu}(\mathbb{Q})+\gamma_f(\mu) \:\forall\mu\in\mathcal{V}\:\forall\:\mathbb{Q}\in\mathcal{Q}\right\rbrace\\
		&=\bigcap_{\mu\in\mathcal{V}}\left\lbrace X\in L^\infty\colon E_\mathbb{Q}[-X]\leq \alpha_{\rho^\mu}(\mathbb{Q})+\gamma_f(\mu) \:\forall\:\mathbb{Q}\in\mathcal{Q}\right\rbrace\\
		&=\bigcap_{\mu\in\mathcal{V}}\left\lbrace X\in L^\infty\colon\rho^\mu(X)\leq\gamma_f(\mu)\right\rbrace=\bigcap_{\mu\in\mathcal{V}}\{\mathcal{A}_{\rho^\mu}-\gamma_f(\mu)\}.
		\end{align*}
		\item This is directly obtained from (i) since from Lemma \ref{lmm:f} in this case we have $\gamma_f(\mu)=0$ if $\mu\in\mathcal{V}_f$ and $\gamma_f(\mu)=\infty$ otherwise. Hence, we get \[	\mathcal{A}_\rho=\bigcap_{\mu\in\mathcal{V}}\{\mathcal{A}_{\rho^\mu}-\gamma_f(\mu)\}=\bigcap_{\mu\in\mathcal{V}_f}\{\mathcal{A}_{\rho^\mu}-\gamma_f(\mu)\}=\bigcap_{\mu\in\mathcal{V}_f}\mathcal{A}_{\rho^\mu}.\]
	\end{enumerate}
\end{proof}

\begin{Rmk}
	It becomes clear that the pivotal role player by $\rho^\mu$ for dual representations of section \ref{sec:main} is also present for acceptance sets. A financial interpretation is that in order for a position $X$ be acceptable for the combination $f(\rho_{\mathcal{I}})$ it must be acceptable for all possible weighting schemes $\mu$ adjusted by a correction, represented by $\gamma_f$. Without such adjustment, the set would be too restrictive. In fact, for $f$ with [PH], we can reduce the restriction to weight schemes over $\mathcal{V}_f$.
\end{Rmk}

\begin{Rmk}
	The results in the last Theorem agree with the four cases of Theorem \ref{thm:dualcomp}. More precisely, if $\rho_{\mathcal{I}}$ is composed as coherent risk measures, then \[\mathcal{A}_\rho=\left\lbrace X\in L^\infty\colon E_\mathbb{Q}[-X]\leq\inf\limits_{\mu\in\mathcal{V}}\left\lbrace \alpha_{\rho^\mu}(\mathbb{Q})+\gamma_f(\mu)\right\rbrace \:\forall\:\mathbb{Q}\in\cup_{\mu\in\mathcal{V}}cl(\mathcal{Q}_{\rho^\mu})\right\rbrace.\] Similar deductions as those for the general convex case lead to $\mathcal{A}_\rho=\bigcap_{\mu\in\mathcal{V}}\{\mathcal{A}_{\rho^\mu}-\gamma_f(\mu)\}$. Furthermore, when  $\rho_{\mathcal{I}}$ is composed as coherent risk measures and $f$ possesses [PH] we get \begin{align*}
	\mathcal{A}_\rho&=\left\lbrace X\in L^\infty\colon E_\mathbb{Q}[-X]\leq 0 \:\forall\:\mathbb{Q}\in clconv(\cup_{\mu\in\mathcal{V}_f}cl(\mathcal{Q}_{\rho^\mu}))\right\rbrace\\
	&=\left\lbrace X\in L^\infty\colon E_\mathbb{Q}[-X]\leq 0 \:\forall\:\mathbb{Q}\in \cup_{\mu\in\mathcal{V}_f}cl(\mathcal{Q}_{\rho^\mu})\right\rbrace\\
	&=\bigcap_{\mu\in\mathcal{V}_f}\left\lbrace X\in L^\infty\colon E_\mathbb{Q}[-X]\leq 0\:\forall\:\mathbb{Q}\in cl(\mathcal{Q}_{\rho^\mu})\right\rbrace\\
	&=\bigcap_{\mu\in\mathcal{V}_f}\left\lbrace X\in L^\infty\colon\rho^\mu(X)\leq0\right\rbrace=\bigcap_{\mu\in\mathcal{V}_f}\mathcal{A}_{\rho^\mu}.
	\end{align*}
\end{Rmk}

\begin{Rmk}
	As examples from the last Theorem, it is worth exploring the particular cases of $\rho^\mu$ and $\rho^{WC}$. For $\rho^\mu$, note that $f(R)=\int_{\mathcal{I}}Rd\mu$ leads to $\gamma_f$ assuming value $0$ in $\mu$ and $\infty$ in $\mathcal{V}\backslash\{\mu\}$. Thus, $\mathcal{V}_f=\{\mu\}$ and we indeed have $\mathcal{A}_{\rho^\mu}=\bigcap_{\nu\in\mathcal{V}_f}\mathcal{A}_{\rho^\nu}=\mathcal{A}_{\rho^\mu}$. Concerning to $\rho^{WC}$, $f(R)=\sup R$ leads to $\gamma_f=0$. Thus, we must have that $\mathcal{A}_{\rho^{WC}}=\bigcap_{i\in\mathcal{I}}\mathcal{A}_{\rho^i}=\bigcap_{\mu\in\mathcal{V}}\mathcal{A}_{\rho^\mu}$. In fact, it is straightforward that $\rho^{WC}(X)\leq 0$ if and only if $\int_{\mathcal{I}}\rho^i(X)d\mu\leq 0$ for any $\mu\in\mathcal{V}$, which corroborates to the claim.
\end{Rmk}

	\bibliography{ref}
	\bibliographystyle{elsarticle-harv}
\end{document}